\documentclass[letterpaper,11pt]{article}
\RequirePackage{amsmath, amsthm, amssymb, mathtools}
\RequirePackage{comment}
\RequirePackage{graphicx}
\RequirePackage{subcaption} 
\RequirePackage{fancyvrb}
\RequirePackage[usenames,dvipsnames,svgnames,table]{xcolor}
    \definecolor{linkcolor}{RGB}{0,0,128}
\RequirePackage[pdfpagelabels,plainpages=false,hypertexnames=true,colorlinks=true,linkcolor=linkcolor,anchorcolor=linkcolor,citecolor=linkcolor,filecolor=linkcolor,menucolor=linkcolor,runcolor=linkcolor,urlcolor=linkcolor,pdfborder={0 0 0}]{hyperref}
\RequirePackage[scale=0.95]{sourcecodepro}
\RequirePackage{mathpazo}

\RequirePackage{dsfont}
\RequirePackage{parskip}
\RequirePackage{tikz}
    \usetikzlibrary{positioning}
\RequirePackage[backend=biber,backref=true,style=numeric,giveninits=true,natbib=true,maxbibnames=12]{biblatex}
\RequirePackage{microtype}
\RequirePackage{forest}

\newcommand*\circled[1]{\tikz[baseline=(char.base)]{
            \node[shape=circle,draw,inner sep=2pt] (char) {#1};}}

\usepackage{stfloats}

\usepackage{newfloat}
\DeclareFloatingEnvironment[
    fileext=los,
    listname={List of Algorithms},
    name=Algorithm,
    placement=tbhp,
    within={},
]{algorithm}

\theoremstyle{plain}
\newtheorem{prototheorem}{Theorem}
\newtheorem{theorem}[prototheorem]{Theorem}
\newtheorem{lemma}[prototheorem]{Lemma}

\newtheorem{corollary}[prototheorem]{Corollary}

\theoremstyle{remark}
\newtheorem{remark}{Remark}

\theoremstyle{definition}

\DeclareMathOperator{\sign}{sgn}
\newcommand{\pos}[2]{#1^{(#2)}}

\DeclarePairedDelimiter\floor{\lfloor}{\rfloor}

\newcommand{\R}{\mathbb{R}}

\title{Incorporating Local Step-Size Adaptivity \\ into the No-U-Turn Sampler using Gibbs Self Tuning}
\author{Nawaf Bou-Rabee\thanks{Department of Mathematical Sciences, Rutgers University, \href{mailto:nawaf.bourabee@rutgers.edu}{\texttt{nawaf.bourabee@rutgers.edu}}}
\and
Bob Carpenter\thanks{Center for Computational Mathematics, Flatiron Institute, \href{mailto:bcarpenter@flatironinstitute.org}{\texttt{bcarpenter@flatironinstitute.org}}}
\and
Tore Selland Kleppe\thanks{Department of Mathematics and Physics, University of Stavanger, 
\href{mailto:mtore.kleppe@uis.no}{\texttt{tore.kleppe@uis.no}}}
\and
Milo Marsden\thanks{Department of Mathematics, Stanford University, 
\href{mailto:mmarsden@stanford.edu}{\texttt{mmarsden@stanford.edu}}}
}

\begin{document}

\maketitle
\begin{abstract}
\noindent
 Adapting the step size locally in the no-U-turn sampler (NUTS) is challenging because the step-size and path-length tuning parameters are interdependent. The determination of an optimal path length  requires a predefined step size, while the ideal step size must account for errors along the selected path. Ensuring reversibility further complicates this tuning problem. In this paper, we present a method for locally adapting the step size in NUTS that is an instance of the Gibbs self-tuning (GIST) framework. Our approach guarantees reversibility with an acceptance probability that depends exclusively on the conditional distribution of the step size.  We validate our step-size-adaptive NUTS method on Neal's funnel density and a high-dimensional normal distribution, demonstrating its effectiveness in challenging scenarios.
 \end{abstract}

\section{Introduction}

Hamiltonian Monte Carlo (HMC), when well-tuned, often generates a Markov chain that mixes rapidly, particularly in high-dimensional spaces \cite{DuKePeRo1987, Ne2011, BePiRoSaSt2013, seiler2014positive,chen2020fast}.  HMC updates the states of a Markov chain by alternating between Gibbs steps of momentum refreshment and Metropolis-within-Gibbs steps with proposals generated by leapfrog simulations of Hamiltonian dynamics.  Path length, a crucial tuning parameter, specifies the number of leapfrog steps taken between consecutive momentum refreshments.    Suboptimal path lengths can result in either highly correlated chain states if too short, or cause undesirable U-turns if too long, both of which hinder mixing, a phenomenon first observed by Mackenzie \cite{Ma1989}.
 
The no-U-turn sampler (NUTS) effectively addresses this path-length tuning problem \cite{HoGe2014,betancourt2017conceptual, carpenter2016stan}.  NUTS is a recursive algorithm that doubles the leapfrog integration time at each step in a random direction (forward or backward in time) until a U-turn (in position) is detected. The next state is then randomly selected from the generated leapfrog iterates, favoring those iterates that are further away from the starting point and those that have lower energy.  This approach ensures the reversibility of the NUTS chain while maintaining computational costs linear in the number of leapfrog steps.  These properties make NUTS the preferred sampler for continuously differentiable densities in probabilistic programming environments \cite{carpenter2016stan,salvatier2016probabilistic,nimble-article:2017,ge2018t,phan2019composable}.
 
However, the effectiveness of path length self-tuning can be undermined without a well-tuned leapfrog step size.   Specifically, the leapfrog step size must be  small enough to keep the energy error  within acceptable bounds during leapfrog integrations.  Poor tuning of the leapfrog step size can lead to large energy errors, creating bottlenecks that can severely impede the mixing of NUTS. These bottlenecks are particularly acute in stiff and high-dimensional problems, which are frequently encountered in practice.

In stiff problems, a smaller leapfrog step size is necessary for stability of the leapfrog integrator.  This step size restriction is due to the fact that the leapfrog integrator is only conditionally stable \cite{LeRe2004};  see Figure 2 of Ref.~\cite{bou2012patch} for an illustration of conditional stability). Thus, if the leapfrog step size is fixed, the smallest necessary step size must be used throughout state space, which can impair mixing in terms of computational cost.  A classic example of this type of bottleneck is Neal’s funnel problem \cite{Neal2003Slice}, where variable curvature causes NUTS to mix slowly \cite{betancourt2013hamiltonianmontecarlohierarchical, modi2023delayed, turok2024sampling}.

In high-dimensional problems, energy errors can accumulate, resulting in bottlenecks even if the leapfrog step size meets the stability requirement. However, due to a concentration phenomenon, these energy errors typically decrease once the chain enters the high probability mass region, or typical set, where the target distribution concentrates.   For instance, in a $d$-dimensional standard normal target, one can prove that the energy error scales as $O(h^2 d)$ outside the typical set, but improves to $O(h^4 d)$ in the typical set, due to the geometric properties of the leapfrog integrator and Gaussian concentration \cite{BePiRoSaSt2013,seiler2014positive}.
 
These challenging scenarios highlight the need for step-size adaptivity in NUTS.  However, locally adapting the leapfrog step size alongside the NUTS approach to adapting the number of steps presents a daunting challenge, primarily due to its recursive architecture and the need to maintain reversibility.  As a result, local adaptation of step size within NUTS is a non-trivial problem.  

This paper introduces a novel solution to this combined step-size and path-length tuning problem, based on the recently introduced Gibbs self-tuning (GIST) framework \cite{BouRabeeCarpenterMarsden2024}.  Our approach treats step-size-adaptive NUTS as a Gibbs sampler in an enlarged space, where the leapfrog step size is considered  a dynamic variable.  Additionally, the Metropolization within GIST corrects for discretization artifacts that prevent exact reversibility, thereby ensuring that reversibility is maintained even as the step size locally adapts to the geometry of the target distribution. Our numerical results demonstrate that incorporating local step-size adaptivity through the GIST framework improves the performance of NUTS on challenging target densities.
 
The paper is structured as follows.  First, we review the GIST framework and show how it can be used to reformulate NUTS as a GIST sampler. Next, we discuss different methods for leapfrog step-size selection, concluding with a GIST-based approach for adapting step-size in NUTS while preserving reversibility. We then demonstrate that our step-size-adaptive NUTS method effectively addresses stiffness in funnel problems and energy accuracy issues in high-dimensional settings.  A detailed proof of the reversibility of step-size-adaptive NUTS method is provided in Appendix~\ref{app:proof}.

\section{Short overview of Gibbs self tuning}

\label{sec:GIST_sampler}

Gibbs self-tuning (GIST) offers a powerful framework for locally tuning parameters within Metropolis samplers, including NUTS \cite{BouRabeeCarpenterMarsden2024}. To understand its potential, let's briefly review the key concepts underpinning HMC samplers from a Gibbs perspective.  Given a  target distribution with (non-normalized) density $e^{-U(\theta)}$ on $\mathbb{R}^d$ where
$U: \mathbb{R}^d \to \mathbb{R}$ is a continuously differentiable potential energy function, HMC samplers  generate a Markov chain on $\mathbb{R}^d$ that leaves this target distribution invariant.  They do this by enlarging state space to include an auxiliary momentum variable  $\rho \in \mathbb{R}^{d}$.  On this enlarged state space $\mathbb{R}^{2d}$, a joint distribution is defined that combines the target distribution with a Gaussian distribution over the momentum variables.  For simplicity, we assume the Gaussian distribution in momentum is a $d$-dimensional standard normal, and hence, the enlarged density is defined by
\begin{equation} \label{eq:target}
f(\theta, \rho) \propto e^{-H(\theta,\ \rho)} \;,  \quad \text{where}~~H(\theta,\rho) = U(\theta) + \frac{1}{2} \, \rho^{\top} \, \rho \;.
\end{equation}
We write the corresponding probability distribution as $\mu$ with background measure on the enlarged space $\mathbb{R}^{2d}$ given by Lebesgue measure $m^{2d}$. 

Given a step size $h>0$, a fundamental component of an HMC sampler is a leapfrog integrator $\Phi_h: \mathbb{R}^{2d} \to \mathbb{R}^{2d}$ for Hamilton's equations,
\begin{equation} \label{eq:exact_flow}
\frac{d}{dt} \theta_t = \rho_t \qquad \textrm{and} \qquad \frac{d}{dt} \rho_t = - \nabla U(\theta_t) \;. 
\end{equation}

\begin{algorithm}[t]
\begin{flushleft}
$\texttt{leapfrog}(\theta, \rho, L, h)$
\vspace*{2pt}
\hrule
\vspace*{2pt}
\textrm{input:}
\begin{tabular}[t]{ll}
$(\theta, \rho) \in \mathbb{R}^{2d}$ & position, momentum \\[2pt] 
$h>0$ &  leapfrog step size  \\[2pt] 
$L \in \mathbb{Z}$ &  number of leapfrog steps \\[2pt] 
\end{tabular} 
\vspace*{4pt}
\hrule
\vspace*{8pt}
$\theta^{(0)} = \theta, \rho^{(0)}= \rho, H^{(0)} = \frac{1}{2} |\rho^{(0)}|^2 + U(\theta^{(0)}) , H^+ = H^- = H^{(0)}$ \hfill (initialize) \\[4pt]
if $L<0$ 
\\[-12pt]
\null \qquad  $\rho^{(0)}= -\rho^{(0)}$ \hfill (reverse for negative steps)
\\[6pt]
for $i$ from $0$ to $|L| - 1$ (inclusive):  \hfill ($L$ leapfrog steps)
\\[-6pt]
\null \qquad $\pos{\rho}{i + 1/2} = \pos{\rho}{i} - \frac{h}{2}  \nabla U(\pos{\theta}{i})$ \hfill (half step momentum)
\\[-6pt]
\null \qquad $\pos{\theta}{i + 1} = \pos{\theta}{i} + h  \pos{\rho}{i + 1/2}$ \hfill (full step position)
\\[-6pt]
\null \qquad $\pos{\rho}{i + 1} = \pos{\rho}{i + 1/2} - \frac{h}{2} \cdot \nabla U(\pos{\theta}{i + 1})$ \hfill (half step momentum)
\\[-6pt]
\null \qquad $H^{(i+1)} = \frac{1}{2} |\rho^{(i+1)}|^2 + U(\theta^{(i+1)})$ \hfill ($(i+1)$-th energy) \\[-6pt]
\null \qquad $H^+ = \max(H^{(i+1)} ,  H^+), H^- = \min(H^{(i+1)}, H^-)$ \hfill (update max/min energy) \\[6pt]
if $L<0$ 
\\[-12pt]
\null \qquad  $\rho^{(L)}= -\rho^{(L)}$ 
\hfill (reset momentum if reversed)
\\[6pt]
 return $ \theta^{(L)}, \rho^{(L)}, H^+, H^-$ \hfill 
\vspace*{4pt}
\hrule
\caption{\it The leapfrog algorithm returns the position, momentum, maximum Hamiltonian and minimum Hamiltonian resulting from $|L|$ leapfrog steps with step size $h$ either forward ($L \ge 0$) or backward ($L < 0$) in time starting from position and momentum $(\theta, \rho)$.}
\label{algo:leapfrog}
\end{flushleft}
\end{algorithm}

Algorithm~\ref{algo:leapfrog} describes $L$ leapfrog steps with step size $h$ from an initial position and momentum $(\theta, \rho) \in \mathbb{R}^{2d}$. While the leapfrog integrator introduces discretization error and does not preserve $\mu$, it remains volume-preserving (i.e., $m^{2d}$-preserving) and reversible.  As a result, the composition $\mathcal{S} \circ \Phi_h$ forms an $m^{2d}$-preserving involution, where $\mathcal{S}: \mathbb{R}^{2d} \to \mathbb{R}^{2d}$ is the momentum flip map defined by $S(\theta, \rho) = (\theta, -\rho)$ \cite{LeRe2004,Hairer2010GeometricNumerical}.  The leapfrog integrator combined with the momentum flip can be used as a proposal move in a Metropolis algorithm.  The Metropolis adjustment ensures detailed balance and yields a correct Markov chain in which the transition kernel leaves the target distribution invariant \cite{BoSaActaN2018}.  

The HMC sampler can be viewed as a Gibbs sampler in this enlarged space alternating between Gibbs momentum refreshments and 
a Metropolis-within-Gibbs step that uses a proposal move defined by $(\theta, \rho) \mapsto \mathcal{S} \circ \Phi_h^L(\theta, \rho)$, where $L$ is the number of leapfrog steps, $\Phi^0_h$ is the identity map with $\Phi^0_h(\theta, \rho) = (\theta, \rho)$, and $\Phi^{L + 1}_h = \Phi_h \circ \Phi^L_h$. By leveraging Hamiltonian flows, well-tuned HMC can traverse high-dimensional spaces more rapidly than traditional MCMC methods, making it well-suited for sampling complex, high-dimensonal distributions.   Note, in this setup, there are two key HMC tuning parameters: the step size $h$ and the number of leapfrog steps $L$.\footnote{In the general case, there is also a positive-definite mass matrix to tune; we are implicitly taking it to be the identity to focus on step-size and path-length tuning.}  

GIST extends HMC by incorporating these tuning parameters as an auxiliary variable denoted by $\alpha$, updated based on the local state of the chain.  We assume $\alpha$ takes values in a space $\mathcal{A}$ with background measure $\nu$.  This measure functions similarly to the reference Lebesgue measure $m^d$ for the auxiliary momentum variable.  On the enlarged space $\mathbb{R}^{2d} \times \mathcal{A}$, a joint density is defined by specifying a conditional density of the tuning parameter given the position and momentum,
\begin{equation} \label{eq:f}
f(\theta, \rho, \alpha) \propto  e^{-H(\theta,\, \rho)} \, p( \alpha  \mid  \theta, \rho)  \;.
\end{equation}
The corresponding joint distribution has density $f$ relative to the background measure on the enlarged space $\mathbb{R}^{2d} \times \mathcal{A}$ given by $m^{2d} \otimes \nu$.
Moreover, for every $(\theta, \rho) \in \mathbb{R}^{2d}$, $p(\alpha \mid \theta, \rho)$ is a conditional density relative to $\nu$.   Analogously to the HMC sampler, a GIST sampler is fully defined by specifying the tuning parameter distribution in \eqref{eq:f} and a measure-preserving involution on the enlarged space,
\begin{equation}
\label{eq:G}
G: (\theta, \rho, \alpha) \mapsto (F(\alpha)(\theta, \rho), \pi(\theta, \rho)(\alpha)) \;,
\end{equation} 
where for each $\alpha \in \mathcal{A}$ the map $F(\alpha): \mathbb{R}^{2d} \to \mathbb{R}^{2d}$ is $m^{2d}$-preserving and for each $(\theta, \rho) \in \mathbb{R}^{2d}$ the map $\pi(\theta, \rho): \mathcal{A} \to \mathcal{A}$ is a $\nu$-preserving involution.  In what follows, $\pi(\theta,\rho)$ will be either the identity function, a sign flip, a shift, or a more complex function related to subtree doubling in NUTS.  Similarly, $F(\alpha)$ is  arbitrary, though it will typically be the result of a simulation of Hamiltonian dynamics in discrete steps using the leapfrog algorithm.

Given the current state of the chain $\theta_0 \in \mathbb{R}^{d}$, the GIST sampler finds an updated state $\theta_1 \in \mathbb{R}^{d}$ by the following steps.
\begin{itemize}
 \item Gibbs momentum refreshment step 
 $\rho_{0} \sim \mathrm{normal}(0, \textrm{I}_{d \times d}) $.
\item Gibbs tuning parameter refreshment step $\alpha_{0}\sim p(\cdot \mid \theta_0,\rho_{0})$.
\item Metropolis-within-Gibbs step with proposal $(\theta_1^*, \rho_1^*, \alpha_1^*) = G(\theta_0, \rho_{0}, \alpha_{0})$ and target $f$ in \eqref{eq:f}, with update
\begin{equation} \label{eq:met_within_gibbs}
(\theta_1, \rho_1, \alpha_1) \ =\ \begin{cases}
                   (\theta_1^*, \rho_1^*, \alpha_1^*) &\text{with probability } a(\theta_0, \rho_{0}, \alpha_{0}), \textrm{ and} \\
                    (\theta_0,\rho_{0},\alpha_{0}) &\mathrm{otherwise.}
                    \end{cases} \end{equation}
                    \end{itemize}
The acceptance probability function $a$ in \eqref{eq:met_within_gibbs} is defined by  
\begin{equation} \label{eq:GISTap}
a(\theta, \rho, \alpha) =  1 \wedge \left( e^{-\Delta H(\theta, \rho)} \, \frac{p( \pi(\theta, \rho)(\alpha) \mid F(\alpha)(\theta, \rho))}{p(\alpha \mid \theta, \rho)} \right) 
 \;,
\end{equation} where $\Delta H(\theta, \rho) := H \circ F(\alpha)(\theta,\rho) - H(\theta,\rho)$.  Theorem~1 in Ref.~\cite{BouRabeeCarpenterMarsden2024} establishes the correctness of GIST samplers, i.e., the transition kernel of the GIST sampler leaves the target distribution invariant.  A GIST sampler is fully specified by the tuning parameter conditional $p(\alpha \mid \theta,\rho)$ in \eqref{eq:f} and the measure-preserving involution $G$ in \eqref{eq:G}.

\section{The no-U-turn sampler as a GIST sampler}
\label{sec:NUTS-fixed-step}



As noted in \cite{BouRabeeCarpenterMarsden2024}, the no-U-turn sampler (NUTS) can be understood as a GIST sampler.  In this section, we provide a self-contained GIST formulation of NUTS by specifying a GIST sampler whose transition step aligns with a NUTS transition step. For any $m,n \in \mathbb{Z}$, let $[m:n] = \{ m, m+1, ..., n \}$, i.e., the set of consecutive integers from $m$ to $n$.  Recall that NUTS generates an index set of leapfrog iterates (or a leapfrog orbit) $[a:b]$ by randomly sampling the endpoints $a, b \in \mathbb{Z}$.  It then selects one of these iterates by sampling its index $L \in [a : b]$. The proposal, $(\theta^*, \rho^*) = \Phi_h^L(\theta, \rho)$, is always accepted.  Algorithm~\ref{algo:nuts} outlines these two crucial substeps of NUTS, which we will describe in greater detail next.  Although the following algorithms are self-contained, they are not optimized for performance. 

\begin{algorithm}[t]
\begin{flushleft}
$\texttt{NUTS}(\theta, h, R, M)$
\vspace*{2pt}
\hrule
\vspace*{2pt}
\textrm{Inputs:}
\begin{tabular}[t]{ll}
$\theta \in \mathbb{R}^{d}$ & position \\[2pt] 
$h>0$ &  coarse step size
\\[2pt] 
$R \ge 1$ &  fine step size is $h/R$
\\[2pt] $M \in \mathbb{N}$ &  maximum size of leapfrog orbit is $2^M$  \\[2pt] 
\end{tabular} 
\vspace*{4pt}
\hrule
\vspace*{8pt}
 $\rho \sim \textrm{normal}(0, \textrm{I}_{d \times d})$ \hfill (complete momentum refreshment)
\\[4pt]
$B \sim  \operatorname{uniform}(\{ 0, 1 \}^M)$ \hfill (symmetric Bernoulli process refreshment)
\\[4pt]
$(a, b, \_) = \texttt{leapfrog-orbit-selection}(\theta, \rho, B, h, R)$ \\ [4pt]
$(\theta^*, \_, \_) = \texttt{leapfrog-index-selection}(\theta, \rho, a, b, h, R)$ \\ [4pt]
 return $  \theta^*$ \hfill 
\vspace*{4pt}
\hrule
\caption{\it Transition step for the no-U-turn sampling (NUTS) algorithm with fixed step size.}
\label{algo:nuts}
\end{flushleft}
\end{algorithm}

Fix the leapfrog step size to be $h>0$ and the maximum size of a NUTS orbit to be $2^M$, where $M \in \mathbb{N}$.  For the subsequent discussion on adapting step size, we describe a version of NUTS which, for a positive integer $R$, uses $\Phi_{h/R}^R$ in place of $\Phi_h$.  This means taking $R$ leapfrog steps of size $h/R$ instead of a single leapfrog step of size $h$. To obtain standard NUTS, set $R =1$ in Algorithm \ref{algo:nuts}.  Consider the following auxiliary variables.
\begin{itemize}
\item
A symmetric Bernoulli process $B \in \{ 0,1 \}^M$ representing the potential random Bernoulli directions used in constructing the leapfrog orbit $[a:b]$.
\item
The number of Bernoulli directions $\ell \in [1:M]$ used to construct the actual leapfrog orbit $[a:b]$.
\item 
The endpoints $a,b \in \mathbb{Z}$ of the leapfrog orbit $[a:b]$.
\item 
The index $L \in [a:b]$ of the selected leapfrog iterate.
\end{itemize}
Thus, the auxiliary variable space is $\{ 0,1 \}^M \times [1:M] \times \mathbb{Z}^3$. Given the position and momentum $(\theta, \rho) \in \mathbb{R}^{2d}$, the conditional distribution of these auxiliary variables is defined by
\begin{equation} \label{eq:NUTS-kernel}
p_{\text{NUTS}}(
B, \ell,  a, b,  L \mid \theta, \rho, h, R) \  = \ \frac{1}{2^M} \, P(\ell, a, b \mid \theta, \rho, B, h, R) \, Q(L \mid \theta, \rho, a, b, h, R) \;.
\end{equation}
The distribution $P(\ell, a, b \mid \theta, \rho, B, h, R)$ is precisely defined in Appendix \ref{app:proof}, and the procedure to sample from it is described below.  Additionally, the distribution $Q(L \mid \theta, \rho, a, b, h, R)$ is further discussed below. The factor $1/2^M$ is the probability of outcome $B$ in the uniform distribution over binary sequences of size $M$ (i.e., $B$ is a symmetric Bernoulli process of size $M$).  As indicated in Algorithm~\ref{algo:nuts},  the auxiliary variable $B$ is refreshed in the same way as momentum.

Given $B$ and $(\theta, \rho)$, the number $\ell$ of Bernoulli directions used  and the endpoints $a,b$ of the leapfrog orbit are determined  by a deterministic recursive procedure that stops when a U-turn or sub-U-turn occurs.  A leapfrog orbit $[a : b]$ starting from $(\theta, \rho)$ has a U-turn if it belongs to the set \begin{equation} \label{eq:uturn} 
\text{U-turn} = \{(a, b, \theta, \rho, h, R) \mid \rho_+ \cdot (\theta_+ - \theta_-) < 0 \ \text{  or  } \ \rho_- \cdot (\theta_+ - \theta_-) < 0 \} \;,
\end{equation}
where $(\theta_-, \rho_-) = \Phi_{h/R}^{Ra}(\theta, \rho)$ and $(\theta_+, \rho_+) = \Phi_{h/R}^{Rb}(\theta, \rho)$. The indicator function $\mathbb{1}_{\text{U-turn}}(a, b, \theta, \rho, h, R)$ equals $1$ if the leapfrog iterates starting from $(\theta, \rho)$ indexed by $[a:b]$ has a $U$-turn.  The sub-U-turn indicator  $\mathbb{1}_{\text{sub-U-turn}}(a, b, \theta, \rho, h, R)$ is defined recursively as follows. \begin{itemize}
    \item For $a \ne b$, 
\begin{equation} \label{eq:subuturn}
\begin{aligned}
&\mathbb{1}_{\text{sub-U-turn}}(a, b, \theta, \rho, h, R)  
\\
&= \max\!\left( \strut \mathbb{1}_{\text{sub-U-turn}}(a, m, \theta, \rho, h, R), \
       \mathbb{1}_{\text{sub-U-turn}}(m+1, b, \theta, \rho, h, R), \
       \mathbb{1}_{\text{U-turn}}(a, b, \theta, \rho, h, R)\right) 
       \end{aligned}
\end{equation}
where $m = \lfloor (a + b)/2 \rfloor$.  
\item
For $a = b$, $\mathbb{1}_{\text{sub-U-turn}}(a, b, \theta, \rho, h, R) = 0$. 
\end{itemize} 
Algorithms~\ref{algo:indicator-U-turn} and~\ref{algo:indicator-sub-U-turn} show how these indicator functions are evaluated.  Critically for efficiency, there is only a linear number of U-turn evaluations rather than the quadratic number that would result from evaluating all pairs of points along a trajectory.  As a result, there may be internal U-turns not captured by the sub-U-turn function.  Algorithm~\ref{algo:leapfrog-orbit-selection} illustrates how to compute the leapfrog orbit endpoints.  Importantly for the subsequent discussion on adapting the step size, each call to $\texttt{leapfrog}$  allows $R$ leapfrog steps with step size $h/R$.  

\begin{algorithm}[t]
\begin{flushleft}
$\texttt{indicator-U-turn}(a, b, \theta, \rho, h, R)$
\vspace*{2pt}
\hrule
\vspace*{2pt}
\textrm{Inputs:}
\begin{tabular}[t]{ll}
$a,b \in \mathbb{Z}$ & leapfrog orbit endpoints
\\[2pt] 
$(\theta, \rho) \in \mathbb{R}^{2d}$ & current position, momentum \\[2pt]
$h>0$ &  coarse step size
\\[2pt] 
$R \ge 1$ &  fine step size is $h/R$
\\[2pt] 
\end{tabular} 
\vspace*{4pt}
\hrule
\vspace*{8pt}
$(\theta^-, \rho^-, H^+, H^-) = \texttt{leapfrog}(\theta, \rho,R a, h/R)$ \hfill ($R a$ leapfrog steps of size $h/R$)
\\[4pt] 
$(\theta^+, \rho^+, \widetilde{H}^+, \widetilde{H}^-) = \texttt{leapfrog}(\theta, \rho, R b, h/R)$ \hfill ($R b$ leapfrog steps of size $h/R$)
\\[4pt] 
$H^+ = \max( \widetilde{H}^+, H^+)$, \ $H^- = \min( \widetilde{H}^-, H^-)$ \\[4pt]
$\textit{Uturn} = \rho^+ \cdot (\theta^+ - \theta^-) <0$ or $\rho^- \cdot (\theta^+ - \theta^-) <0$
\\[4pt]
 return $ \textit{Uturn}, H^+, H^-$ \hfill 
\vspace*{4pt}
\hrule
\caption{\it Calculation of the indicator function for U-turns.}
\label{algo:indicator-U-turn}
\end{flushleft}
\end{algorithm}

\begin{algorithm}[t]
\begin{flushleft}
$\texttt{indicator-sub-U-turn}(a, b, \theta, \rho, h, R)$
\vspace*{2pt}
\hrule
\vspace*{2pt}
\textrm{Inputs:}
\begin{tabular}[t]{ll}
$a,b \in \mathbb{Z}$ & leapfrog orbit endpoints
\\[2pt]  
$(\theta, \rho) \in \mathbb{R}^{2d}$ & current position, momentum \\[2pt]
$h>0$ &  coarse step size
\\[2pt] 
$R \ge 1$ &  fine step size is $h/R$
\\[2pt] 
\end{tabular} 
\vspace*{4pt}
\hrule
\vspace*{8pt}
if $a=b$:  \hfill 
\\[-12pt]
\null \qquad   $\textrm{return } 0$
\\[4pt]
$m=\lfloor (a+b)/2 \rfloor$
\\[4pt]
$(\textit{full}, \_, \_) = \texttt{indicator-U-turn}(a,b,\theta,\rho,h,R)$ 
\\[-6pt]
\null $\textit{left} = \texttt{indicator-sub-U-turn}(a,m,\theta,\rho,h,R)$
\\[-6pt]
\null  $\textit{right} = \texttt{indicator-sub-U-turn}(m+1,b,\theta,\rho,h,R)$ 
\\[4pt]
return $\max(\textit{full},\textit{left},\textit{right})$ \hfill 
\vspace*{4pt}
\hrule
\caption{\it Calculation of the indicator function for sub-U-turns.}
\label{algo:indicator-sub-U-turn}
\end{flushleft}
\end{algorithm}

Given the leapfrog orbit $[a:b]$, the index $L \in [a:b]$ is randomly sampled from the index selection kernel $Q(L \mid \theta, \rho,a,b, h, R)$.  Typically, $Q$ is taken to be either the categorical distribution based on the Boltzmann weights of the leapfrog iterates, 
\begin{equation} \label{eq:categorical}
Q(L  \mid \theta, \rho, a, b, h, R) 
= 
\dfrac{e^{-H \, \circ \, \Phi_{h/R}^{RL}(\theta, \rho)}}
      {\sum_{k \in [a:b]} e^{-H \, \circ \, \Phi_{h/R}^{Rk}(\theta, \rho)}} \, \mathbb{1}_{[a:b]}(L) \;, 
\end{equation}  or the biased progressive sampling kernel described in \cite{betancourt2017conceptual}.  Algorithm~\ref{algo:leapfrog-index-selection} shows how to sample this categorical distribution with a single pass through the leapfrog iterates.   

To complete the specification of NUTS as a GIST sampler, we define the following measure-preserving involution on the enlarged space,
\begin{equation}
\label{eq:G_nuts} 
G: (\theta, \rho, B, \ell, a, b, L) \mapsto (\theta^*, \rho^*, B^*, \ell, a^*,  b^*, -L) \;,
\end{equation}
where $(\theta^*, \rho^*) = \Phi^{RL}_{h/R}(\theta, \rho)$, $a^* = a-L$, and $b^* = b-L$. The sequence $B^* \in \{ 0, 1\}^M$ is defined such that 
\[
(a^*, b^*, \_) = \texttt{leapfrog-orbit-selection}(\theta^*, \rho^*, B^*, h, R) \;.
\]

The functions \texttt{leapfrog-orbit-selection} and \texttt{leapfrog-index-selection} are defined in Algorithms \ref{algo:leapfrog-orbit-selection} and \ref{algo:leapfrog-index-selection}. As the NUTS proposal is always accepted, the sequence $B^*$ plays no role in the implementation of NUTS and its precise definition is postponed to the next section.

\begin{algorithm}[t]
\begin{flushleft}
$\texttt{leapfrog-orbit-selection}(\theta, \rho, B, h, R)$
\vspace*{2pt}
\hrule
\vspace*{2pt}
\textrm{Inputs:}
\begin{tabular}[t]{ll}
$(\theta, \rho) \in \mathbb{R}^{2d}$ & current position, momentum \\[2pt]
$B$ & binary string of size $M \in \mathbb{N}$
\\[2pt] 
$h>0$ &  coarse step size
\\[2pt] 
$R \ge 1$ &  fine step size is $h/R$
\\[2pt] 
\end{tabular} 
\vspace*{4pt}
\hrule
\vspace*{8pt}
$a = b = 0 $, $H^+ = H^- = \frac{1}{2} |\rho|^2 + U(\theta)$ \hfill (initialize endpoints of index-set and extremal energies)
\\[4pt]
for $i$ from $1$ to $M$ (inclusive):  \hfill (no more than $M$ doublings)
\\[-6pt]
\null \qquad $\widetilde{a} = a + (-1)^{B_{i}} \, 2^{i-1}$, $\widetilde{b} = b + (-1)^{B_{i}} \, 2^{i-1}$  \hfill (endpoints of proposed extension) 
\\[-6pt]
\null \qquad $(\textit{Uturn},\widetilde{H}^+, \widetilde{H}^-) = \texttt{indicator-U-turn}(a,b,\theta, \rho, h, R)$,   \hfill (U-turn?)
\\[-6pt]
\null \qquad $H^+ = \max( \widetilde{H}^+, H^+)$, $H^- = \min( \widetilde{H}^-, H^-)$ \\[-6pt]
\null \qquad $\textit{subUturn} = \texttt{indicator-sub-U-turn}(\widetilde{a},\widetilde{b},\theta, \rho, h, R)$  \hfill (sub-U-turn?)
\\[-6pt]
\null \qquad if $\max(\textit{Uturn}, \textit{subUturn}) = 0$ \hfill (if no-u-turn) 
\\[-6pt]
\null \qquad \qquad  $a=\min(a,\widetilde{a})$,  $b=\max(b,\widetilde{b})$  \hfill (double size of index-set) 
\\[-12pt]
\null \qquad else 
\\[-12pt]
\null \qquad \qquad break 
\\[4pt]
$\Delta H = H^+ - H^-$
\\[4pt]
 return $ a, b,  \Delta H$ \hfill 
\vspace*{4pt}
\hrule
\caption{\it The no-U-turn orbit selection procedure.}
\label{algo:leapfrog-orbit-selection}
\end{flushleft}
\end{algorithm}

\begin{algorithm}[t]
\begin{flushleft}
$\texttt{leapfrog-index-selection}(\theta, \rho, a, b, h, R)$
\vspace*{2pt}
\hrule
\vspace*{2pt}
\textrm{Inputs:}
\begin{tabular}[t]{ll}
$(\theta, \rho) \in \mathbb{R}^{2d}$ & current position, momentum \\[2pt]
$a \in \mathbb{Z}$ & left endpoint of index set 
\\[2pt] 
$b \in \mathbb{Z}$ & right endpoint of index set 
\\[2pt] 
$h>0$ &  coarse step size
\\[2pt] 
$R \ge 1$ &  fine step size is $h/R$
\\[2pt] 
\end{tabular} 
\vspace*{4pt}
\hrule
\vspace*{8pt}
$(\widetilde{\theta}^*, \widetilde{\rho}^*,\_,\_) = \text{leapfrog}(\theta, \rho,R a, h/R)$ \hfill ($R a$ leapfrog steps of size $h/R$)
\\[4pt] 
$\widetilde{H}^* = \frac{1}{2} |\widetilde{\rho}^*|^2 + U(\widetilde{\theta}^*)$,  $w=e^{-\widetilde{H}^*}$ , $L=a$, $(\theta^*, \rho^*) = (\widetilde{\theta}^*, \widetilde{\rho}^*)$ \hfill (initialize)
\\[4pt]
for $i$ from $a+1$ to $b$ (inclusive):  \hfill 
\\[-6pt]
\null \qquad $(\widetilde{\theta}^*, \widetilde{\rho}^*,\_,\_) = \text{leapfrog}(\widetilde{\theta}^*, \widetilde{\rho}^*,R, h/R)$  \hfill ($R i$ leapfrog steps of size $h/R$)
\\[-6pt]
\null \qquad $\widetilde{H}^* = \frac{1}{2} |\widetilde{\rho}^*|^2 + U(\widetilde{\theta}^*)$, 
\ $w = w + e^{-\widetilde{H}^*} $, 
\ $\textit{ap} = e^{-\widetilde{H}^*}/w$ \hfill
\\[-6pt]
\null \qquad  $u \sim \textrm{uniform}([0, 1])$  \hfill 
\\[-6pt]
\null \qquad if $u \le \textit{ap}$ \hfill 
\\[-6pt]
\null \qquad \qquad  $L = i$ \\[-6pt]
\null \qquad \qquad   $(\theta^*, \rho^*) = (\widetilde{\theta}^*, \widetilde{\rho}^*)$ \hfill 
\\[4pt]
 return $  \theta^*, \rho^*, L$ \hfill 
\vspace*{6pt}
\hrule
\caption{\it The leapfrog index selection procedure for the no-U-turn sampler.}
\label{algo:leapfrog-index-selection}
\end{flushleft}
\end{algorithm}

\section{Locally adaptive step-size selection for NUTS}
\label{sec:adaptNUTS}

Locally adapting the step size  in NUTS is a challenging task due to the interdependence between step-size and path-length tuning parameters. Determining the optimal path length requires setting a step size in advance, while the ideal step size must account for the leapfrog energy error incurred along a given path. Additionally, reversibility is essential.  In this section, we introduce a locally adaptive step-size selection method for NUTS using the GIST framework, which simultaneously selects both tuning parameters such that the GIST acceptance probability depends solely on the conditional distribution of the step size. Later, in Sections~\ref{sec:funnel} and~\ref{sec:highDnormal}, we validate this approach on Neal’s funnel density and a high-dimensional standard normal target, both difficult for samplers with fixed step sizes, including NUTS. These numerical results indicate that this adaptive strategy significantly enhances NUTS’s sampling efficiency in these challenging scenarios.

\subsection{Initial point symmetry of step-size selection}

\label{sec:init_point_symm}

A simple, ideal step-size selection method would ensure that the optimal step size for a given $(\theta, \rho)$  aligns with the optimal step size for any point along the leapfrog path passing through $(\theta, \rho)$.  Inserting a step-size function with this initial point symmetry into any reversible GIST sampler that locally adapts path length --- including NUTS --- preserves reversibility. Specifically, for an acceptance probability threshold $a_{\min} \in (0,1)$ and any position and momentum $(\theta, \rho) \in \mathbb{R}^{2d}$, a step-size selection function $\eta(a_{\min}): \mathbb{R}^{2d} \to \mathbb{R}_{>0}$ has \emph{initial point symmetry} if
\begin{equation} \label{eq:step-size-initial-point-symmetry}
h^*=\eta(a_{\min})(\theta, \rho)  \quad \text{satisfies} \quad h^* = \eta(a_{\min})(\Phi^i_{h^*}(\theta, \rho)) \quad  \text{for all $i \in \mathbb{Z}$} \;. 
\end{equation} Here, $\Phi_{h^*}^i(\theta, \rho)$ denotes the leapfrog iterate after $i$ steps  with step size $h^*$ starting from $(\theta, \rho)$.

One idea for a step-size selection function with this property is: \
\begin{equation} \label{eq:eta}
\eta(a_{\min})(\theta, \rho) = \sup \left\{ h>0 \mid  e^{ - (H^+  -  H^-) } \ge a_{\min} \right\} \;,
\end{equation}
where 
\[
H^+ = \max_{ i \in \mathbb{Z}} H \circ \Phi_{h}^i(\theta, \rho)
\quad
\textrm{ and }
\quad
H^- =  \min_{i \in \mathbb{Z}}  H \circ \Phi_{h}^i(\theta, \rho) \;.
\]
In this definition, $H^+$ and $H^-$ represent the maximum and minimum value of the Hamiltonian function along the infinite two-sided leapfrog path passing through $(\theta, \rho)$.  These values are properties of the entire leapfrog path and are therefore uniform with respect to the initial point along the path.  Under mild regularity conditions on the potential energy function, $H^+$ and $H^-$ can be expected to vary continuously with the step size.  Intuitively, the leapfrog iterates lie on a level set of a modified Hamiltonian function \cite{Mo1968}, and the size of the energy error envelope on this level set changes continuously with the step size.  Consequently, for any acceptance probability threshold $a_{\min} \in (0,1)$,  there may exist a unique step size satisfying \eqref{eq:eta}, regardless of the initial point taken along the leapfrog path.  However, the next remark shows that this might not always be the case even for a standard normal target distribution.  

\begin{figure}
\centering
\includegraphics[scale=0.3]{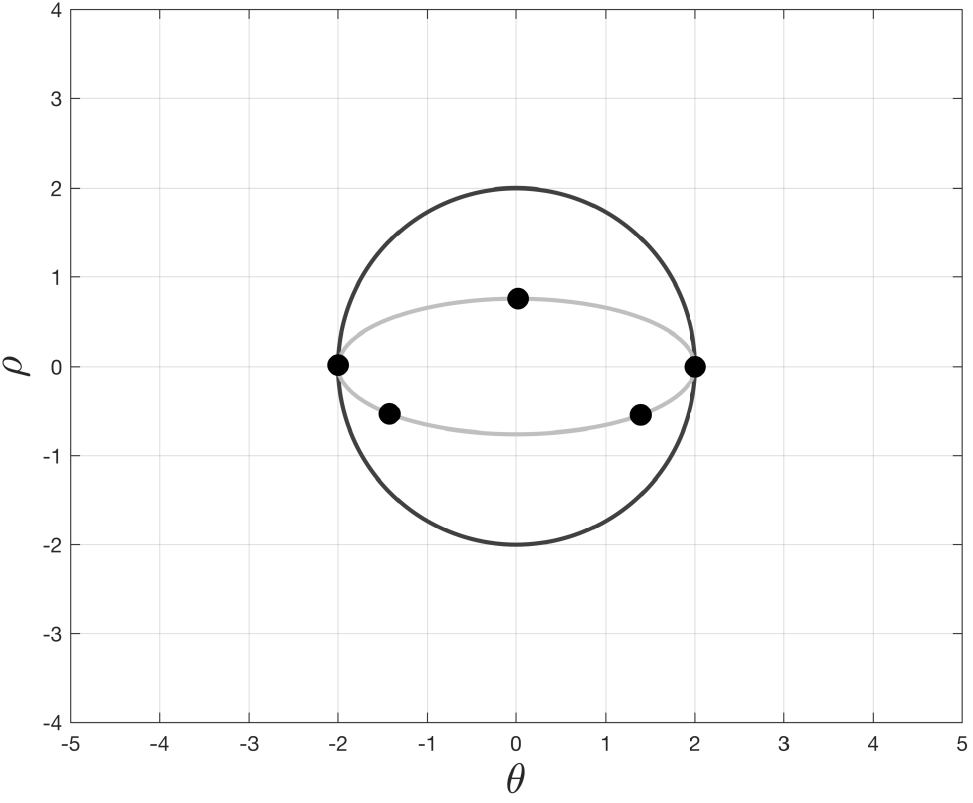}
\includegraphics[scale=0.3]{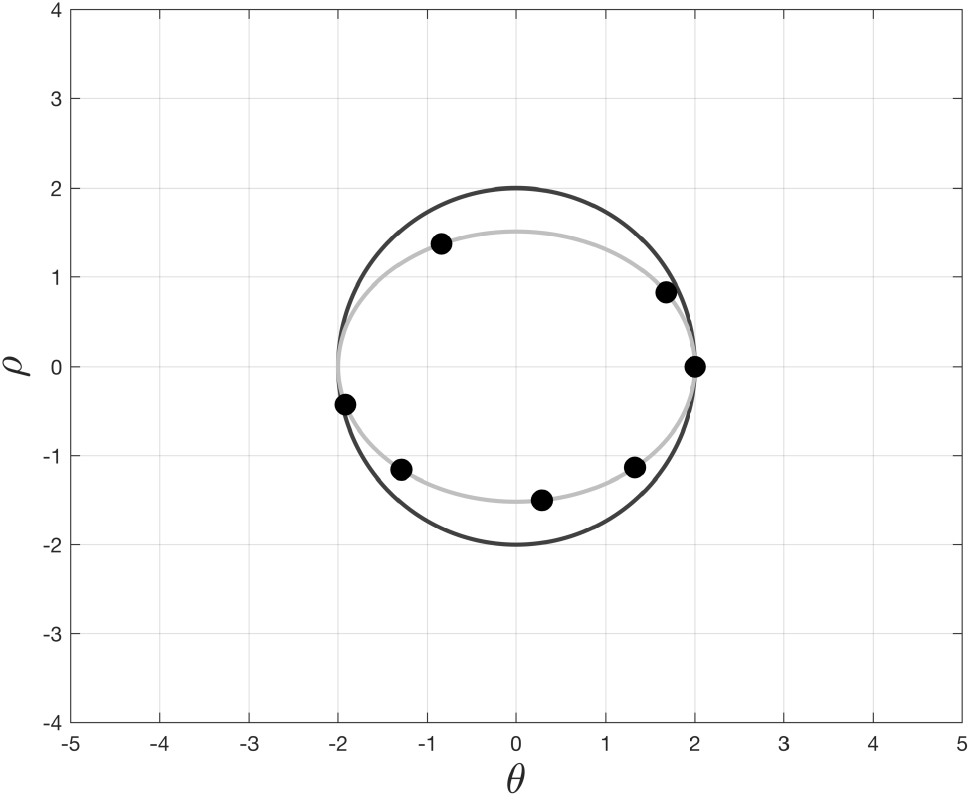}
\includegraphics[scale=0.3]{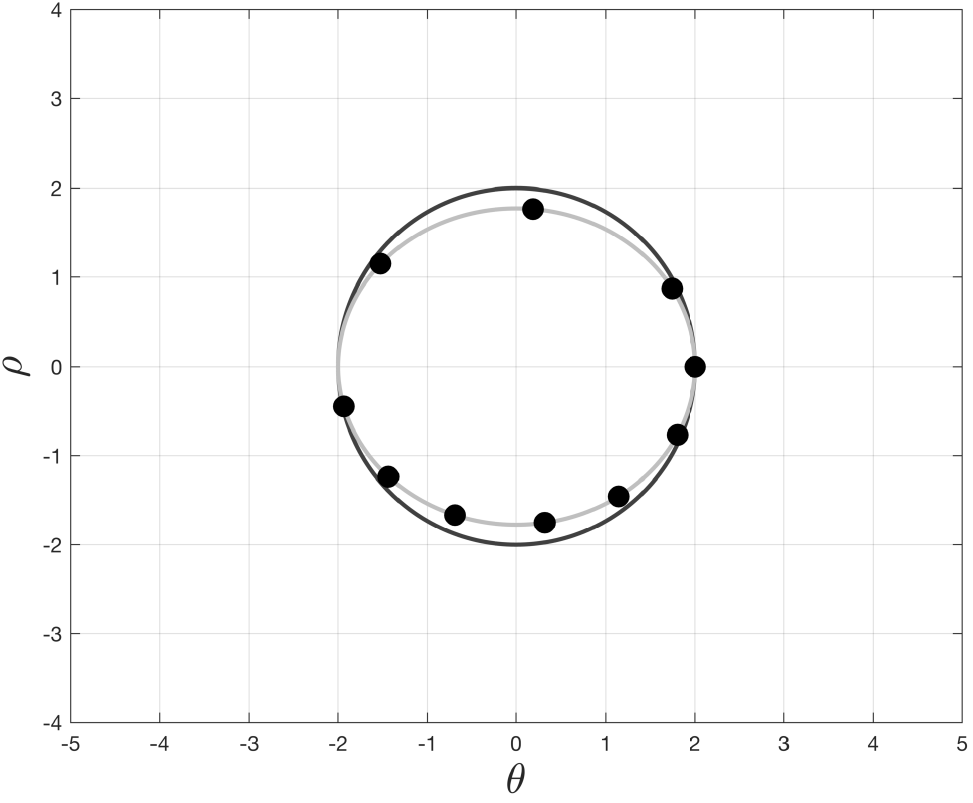} \\
\includegraphics[scale=0.3]{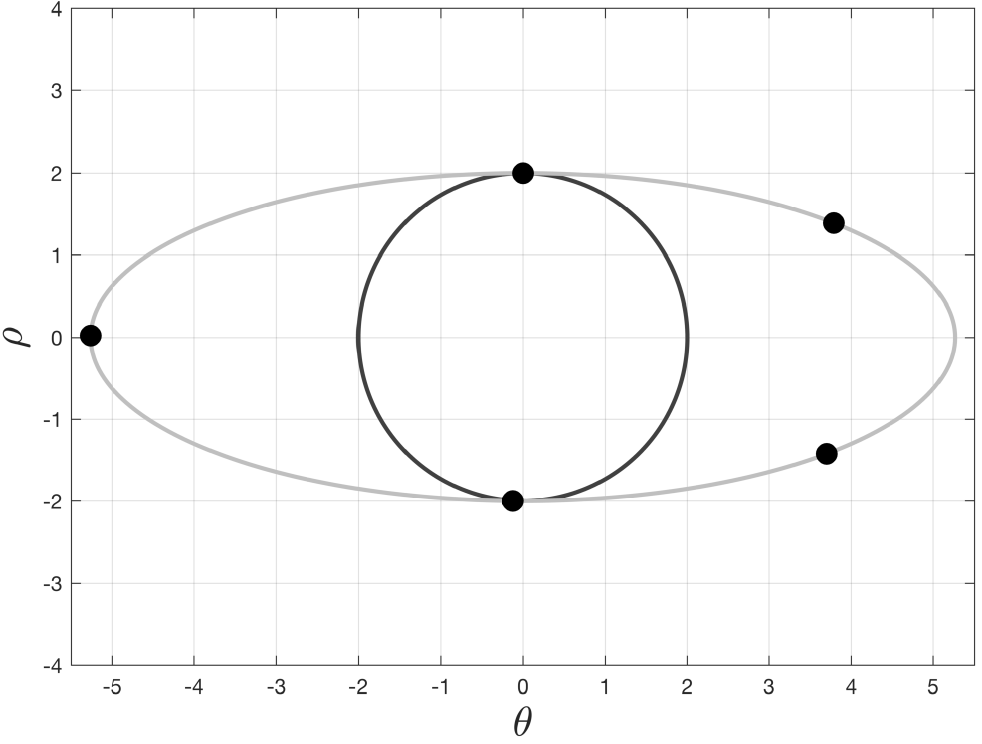}
\includegraphics[scale=0.3]{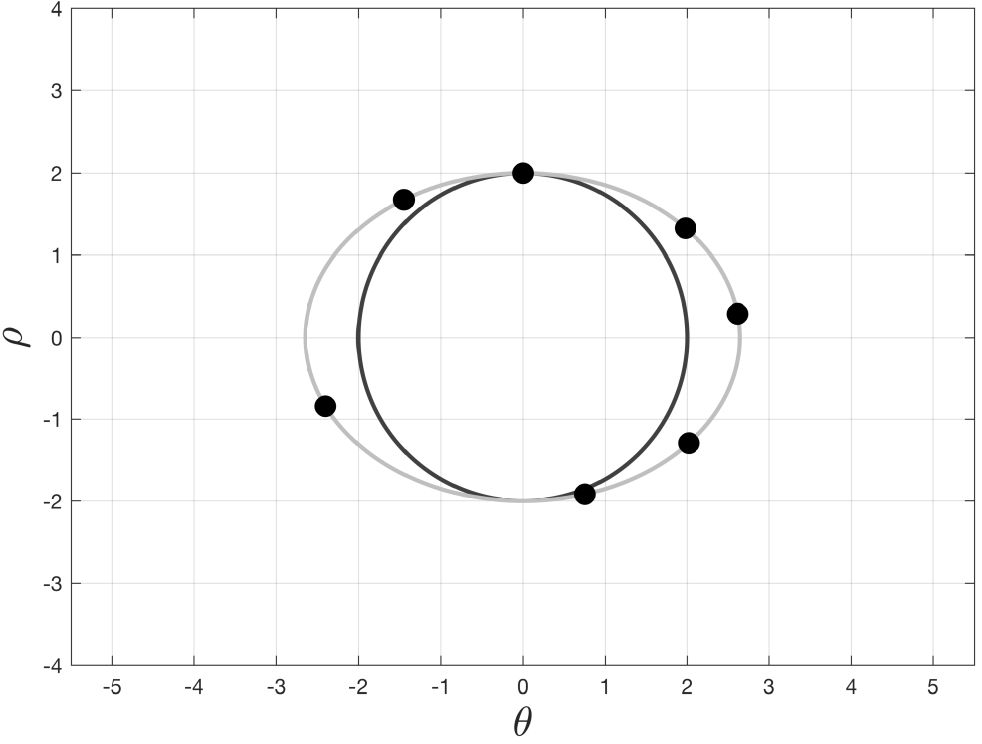}
\includegraphics[scale=0.3]{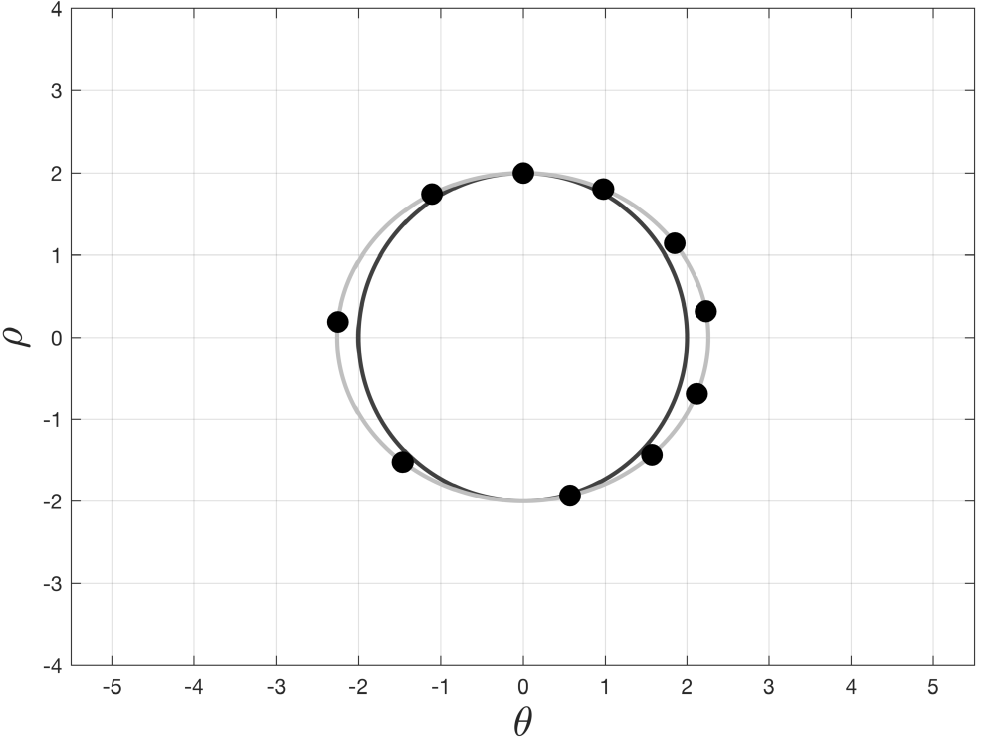}
\caption{\textit{For a standard normal target, the underlying leapfrog integrator preserves a modified Hamiltonian (grey) illustrated here for step sizes $h\in \{1.85, 1.31, 0.925\}$ and initial conditions $(\theta_0,\rho_0) \in \{ (2,0), (0,2) \}$.  In black is a level set of the exact Hamiltonian through $(\theta_0,\rho_0)$.  The dots are points on a leapfrog path with initial condition $(\theta_0,\rho_0)$.}}
\label{fig:stdnormal}
\end{figure}
\begin{remark} 
\label{rmk:modifiedHnormal}
Suppose $h \in (0,2)$.  As illustrated in Figure~\ref{fig:stdnormal}, 
for a standard normal target distribution, for which $H(\theta, \rho) = \frac{1}{2} (|\theta|^2 + |\rho|^2)$, the leapfrog integrator preserves the modified Hamiltonian: 
\[
H_h(\theta, \rho) = \frac{1}{2} \left( 1 - \frac{h^2}{4} \right) |\theta|^2 + \frac{1}{2} |\rho|^2 \;,
\] 
which means $H_h \circ \Phi_h = H_h$.  The level sets of $H_h$ are ellipses with major axes aligned with the $\theta$-axis. Thus, the maximum value of $H$ occurs at $H(\theta^*,0)$ where $\theta^*$ satisfies $H_h(\theta^*,0) = H_h(\theta, \rho)$, and the  minimum value of $H$ occurs at $(0,\rho^*)$ where $\rho^*$ satisfies $H_h(0,\rho^*) = H_h(\theta, \rho)$.  Therefore, $H$ evaluated along any leapfrog path with initial condition $(\theta, \rho)$ satisfies 
\[
\frac{1}{2} \left( \left( 1- \frac{h^2}{4} \right) |\theta|^2 + |\rho|^2 \right) \,  \le \,  H \circ \Phi_h^i(\theta, \rho) \, \le \,  \frac{1}{2} \left( |\theta|^2 + \left(1-\frac{h^2}{4} \right)^{-1} |\rho|^2 \right)
\] 
uniformly in $i \in \mathbb{Z}$.  The difference between these bounds yields
\begin{align}
\label{eq:dH_stdnormal}
(H^+ - H^-)(\theta, \rho) &\le \frac{h^2}{2} \left( \frac{1}{4} |\theta|^2 + \frac{1}{4 - h^2} |\rho|^2 \right) =  \frac{h^2}{ 4-h^2} H_h(\theta, \rho) \;.  
\end{align}
Because $H_h \circ \Phi_h = H_h$, for all $(\theta, \rho) \in R^{2d}$ and $i \in \mathbb{Z}$, 
\[
(H^+ - H^-) ( \Phi_h^i(\theta, \rho) ) \le \frac{h^2}{4 - h^2} H_h \circ \Phi_h^i(\theta, \rho) = \frac{h^2}{ 4-h^2} H_h(\theta, \rho) \; .
\]

\begin{figure}[t]
\centering
\includegraphics[scale=0.45]{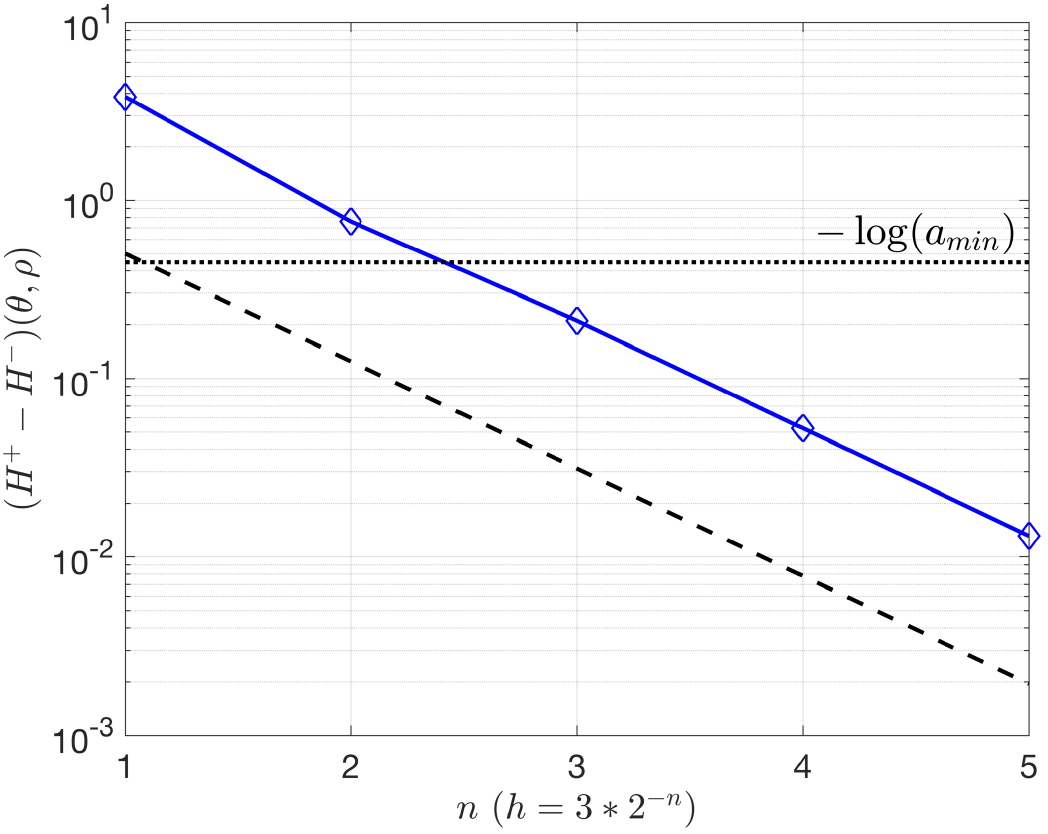} 
\includegraphics[scale=0.45]{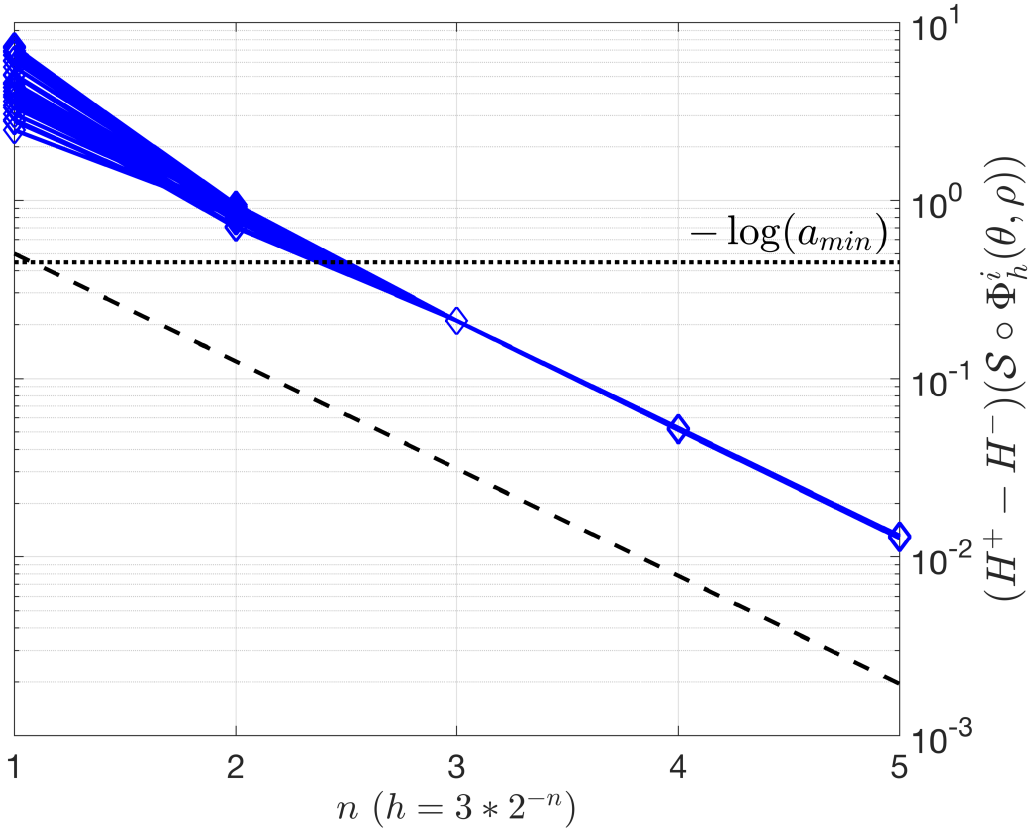}
\caption{\textit{The left panel shows a plot of the path-wise energy error for leapfrog initiated at $(\theta,\rho)=(-3.24,-1.20)$. The dotted line indicates the corresponding acceptance probability threshold $a_{min}=64\%$, which is met at $h=3/8$.  For that value of $h$, the right panel shows the path-wise energy error for initial points $\{ \mathcal{S} \circ \Phi_h^i(\theta, \rho) \}$.  The acceptance probability threshold is met at $h=3/8$ uniformly in $i$.}}
\label{fig:stdnormal2}
\end{figure}

However, the upper bound \eqref{eq:dH_stdnormal} does not ensure that the corresponding step-size selection function maintains initial point symmetry.    Figure~\ref{fig:stdnormal2} illustrates the challenge of ensuring initial point symmetry, even in this simple setting, due to discretization artifacts. This difficulty highlights the importance of using the Metropolis-within-Gibbs step in the GIST sampler to ensure reversiblity, a point that will be further expanded in the following subsection.
\end{remark}

\subsection{Near initial point symmetry of step-size selection} \label{sec:implementable-step-size-selection}

Consider the step-size selection function $\eta$ defined in \eqref{eq:eta}.
By incorporating $a_{min} \in (0,1)$ as an input parameter and including the line $h=\eta(a_{min})(\theta,\rho)$ in Algorithm~\ref{algo:nuts}, we obtain a NUTS sampler with locally adaptive step size.  For this sampler to be reversible, $\eta(a_{min})(\theta,\rho)$ must satisfy the initial point symmetry condition given in \eqref{eq:step-size-initial-point-symmetry}. However, as discussed in Remark~\ref{rmk:modifiedHnormal}, $\eta$ might not always meet this requirement.  Moreover, computing $\eta(a_{min})(\theta,\rho)$ requires searching over a continuous range of step sizes $h>0$ and an infinite set of leapfrog iterates $\{ \Phi_h^i(\theta, \rho) \}_{i \in \mathbb{Z}}$, making it impractical to evaluate. 

To practically implement locally adaptive step-size selection functions, we  approximate $\eta$ by limiting the search to a discrete set of step sizes and a finite leapfrog path.  Since these approximations will not exactly satisfy the initial point symmetry condition \eqref{eq:step-size-initial-point-symmetry},  a Metropolis accept-reject step is used to maintain reversibility.  This step is straightforward to incorporate in the GIST framework, since it allows for a Metropolis-within-Gibbs step as defined in \eqref{eq:met_within_gibbs}. The corresponding GIST acceptance probability will reflect how closely the approximate step size selection function adheres to initial point symmetry. Therefore, good step-size selection functions will approximately satisfy the initial point symmetry condition. 

For instance, fixing a maximum step-size $h >0$ and a time interval $[-T,T]$ where $T>0$, an initial guess at such an implementable step-size selection function might be
\begin{equation*}
\widetilde{\eta}(a_{min}, T)(\theta, \rho) = \max  \left\{ 2^{-k} h \mid e^{-(H^+ - H^-)} \geq a_{min} \right\} \; , 
\end{equation*}
where
\begin{equation*}
H^+ = \underset{-T \, \leq \, ih\, \leq \, T}{\max} H \circ \Phi_h^i(\theta, \rho)
\quad \textrm{and}  \quad 
H^- = \underset{-T \, \leq \, ih \, \leq \, T}{\min} H \circ \Phi_h^i(\theta, \rho) \; .
\end{equation*}  
In essence, $\widetilde{\eta}$ selects the largest step size in the decreasing sequence $\{ 2^{-k} \, h \}_{k \in \mathbb{N}}$ such that the energy envelope from initial condition $(\theta, \rho)$ over the time interval $[-T,T]$ is within the specified tolerance.

However, selecting a step size based solely on $\widetilde{\eta}$ can lead to low GIST acceptance rates, particularly for paths that pass through regions with different energy errors. For a given point $(\theta, \rho)$,  $\widetilde{\eta}$ will select a step size $\tilde{h}$ such that the discrete path  $\{ \Phi_{\tilde{h}}^i(\theta, \rho) \}_{-T \, \leq \, i \tilde{h} \, \leq \, T}$ closely approximates the continuous path $\{\varphi_t(\theta, \rho) \}_{-T \, \leq \, t \, \leq \, T}$. However, for two different points $(\theta, \rho)$ and $(\theta^*, \rho^*)$, their respective continuous paths may differ. A step size that accurately approximates one path  may not work for the other, especially if one path segment lies in a region of low energy error while the other lies in a region of high energy error. 

Moreover, a step size that accurately approximates $\{\varphi_t(\theta, \rho) \}_{-T \, \leq \, t \, \leq \, T}$ may not work well for generating the leapfrog iterates $\{ \Phi_{\tilde{h}}^i(\theta, \rho) \}_{i \in [a : b]}$ used in NUTS, as these two sets of iterates can have different energy errors.   A robust step-size selection procedure for NUTS should ensure near initial point symmetry. Specifically, if a step size $\tilde{h}$ is chosen based on $(\theta, \rho)$, and $(\theta^*, \rho^*)$ is a point in the set $\{ \Phi_{\tilde{h}}^i(\theta, \rho) \}_{i \in [a : b]}$, then the step-size search starting from $(\theta^*, \rho^*)$  should yield a step size similar to $\tilde{h}$.  To address the potential initial point asymmetry in this step-size selection, we introduce randomness by carefully designing a step-size tuning parameter distribution that aligns with the NUTS path-length selection procedure.  This local step-size adaptivity should also maintain the reversibility of the NUTS path-length selection procedure.  The GIST framework provides a unified approach to achieve this, as described next. 

\subsection{Step-size-adaptive NUTS}

Let $h>0$ be a coarse step size, $a_{min} \in (0,1)$ be the acceptance probability threshold, and $B$ be a binary sequence of fixed size $M$.  We define the function:  
\begin{equation} \label{eq:step-size-NUTS}
\texttt{step-reduction}(\theta, \rho, B, h, a_{\textrm{min}}) = \min\{ k \in \mathbb{N} \mid e^{-\Delta H_{\text{gap}}} \geq a_{min} \}
\end{equation}
where $(\_, \_, \Delta H_{\text{gap}}) = \texttt{leapfrog-orbit-selection}(\theta, \rho, B, h, 2^k)$.\footnote{The definition of $\Delta H_{\text{gap}}$ includes the fine gridpoints corresponding to the iterates of $\Phi_{2^{-k}h}$ between the coarse gridpoints defined by the iterates of $\Phi_{2^{-k} h}^{2^k}$.} To select the leapfrog orbit, this algorithm uses $\Phi_{2^{-k} h}^{2^k h}$ in place of $\Phi_h$.  This means that $2^k$ leapfrog steps with a fine step size $2^{-k} h$ are taken instead of a single leapfrog step of size $h$. The coarse step size $h$ remains constant regardless of the value of $k$.  Algorithm~\ref{algo:step-reduction} implements this function.

Using this definition within the GIST framework from Section~\ref{sec:GIST_sampler}, a NUTS sampler with local step-size adaptation can be defined by introducing the tuning parameters \[
(k, B, \ell, a, b, L) \in \mathbb{N} \times \{0,1\}^M \times [1:M] \times \mathbb{Z}^3 \;.
\] This set of tuning parameters is the same as in fixed step-size NUTS, but augmented with $k \in \mathbb{N}$.  The conditional distribution of these tuning parameters is defined as:
\begin{equation} 
\label{eq:adapt-nuts-conditional}
\begin{aligned}
& p_{\textrm{adaptNUTS}}( k, B, \ell, a, b, L \mid \theta, \rho) \\
& \quad =  p_k(k \mid \texttt{step-reduction}(\theta, \rho, B, h, a_{\textrm{min}})) \,  p_{\textrm{NUTS}}(B, \ell, a, b, L \mid \theta, \rho, h, 2^k) \;,
\end{aligned}
\end{equation}
where $p_k(k \mid \texttt{step-reduction}(\theta, \rho, B, h, a_{\textrm{min}}))$ is the conditional distribution of $k$ based on the output of $\texttt{step-reduction}$.  Note, $p_{\textrm{adaptNUTS}}$ differs from $ p_{\textrm{NUTS}}$ only by the factor $p_k$.  The measure-preserving involution in the enlarged space is defined as: 
\begin{equation} \label{eq:adapt-nuts-involution}
G: (\theta, \rho, k, B, \ell, a, b,  L) \mapsto (\theta^*, \rho^*, k, B^*, \ell, a-L, b-L, -L) \;,
\end{equation}
where $(\theta^*, \rho^*) = \Phi_{2^{-k}h}^{2^k L}(\theta, \rho)$ and $B^*$ is defined below in Section~\ref{sec:Bstar}.  While the proposal resides in the coarse grid defined by $h$, the algorithm uses $2^k$ leapfrog steps with a fine step size of $2^{-k} h$ to populate this coarse grid.  The  GIST acceptance probability in \eqref{eq:GISTap} is then: 
\begin{equation} \label{eq:step-size-acceptance-ratio}
 1 \wedge \frac{ p_k(k \mid \texttt{step-reduction}(\theta^*, \rho^*, B^*, h, a_{\textrm{min}}))}{p_k(k \mid \texttt{step-reduction}(\theta, \rho, B, h, a_{\textrm{min}}))} \;.
\end{equation} 

For the choice 
\[
p_k(k \mid \texttt{step-reduction}(\theta, \rho, B, h, a_{\textrm{min}}) = \widetilde{k}) 
= 
\frac{1}{3} \mathbb{1}_{ \{ \widetilde{k}-1, \ \widetilde{k}, \ \widetilde{k}+1 \} } (k) \;,
\] 
Algorithm~\ref{algo:adaptive-NUTS} describes the corresponding step-size-adaptive NUTS sampler.  It uses the function \texttt{step-reduction}, which is defined in Algorithm~\ref{algo:step-reduction}.  The choice of using just two surrounding points involves a tradeoff between wanting to stay close to an optimal step size while maintaining a high Metropolis-within-Gibbs acceptance probability.

A complete proof of the correctness of Algorithm~\ref{algo:adaptive-NUTS} is provided in Appendix~\ref{app:proof}.

\begin{algorithm}[t]
\begin{flushleft}
$\texttt{adaptNUTS}(\theta, h, M, a_{min})$
\vspace*{2pt}
\hrule
\vspace*{2pt}
\textrm{Inputs:}
\begin{tabular}[t]{ll}
$\theta \in \mathbb{R}^{d}$ & current position \\[2pt] 
$h>0$ &  coarse step size  \\[2pt] 
$M \in \mathbb{N}$ &  maximum size of leapfrog orbit is $2^M$  \\[2pt] 
$a_{\min} \in (0,1)$ & acceptance probability threshold \\[2pt] 
\end{tabular} 
\vspace*{4pt}
\hrule
\vspace*{8pt}
 $\rho \sim \textrm{normal}(0, \textrm{I}_{d \times d})$ \hfill (complete momentum refreshment)
\\[4pt]
$B \sim  \operatorname{uniform}(\{ 0, 1 \}^M)$ \hfill (symmetric Bernoulli process refreshment)
\\[4pt]
$\widetilde{k} = \texttt{step-reduction}(\theta, \rho, B, h, a_{\textrm{min}})$
\\[4pt]
$k \sim \operatorname{uniform}(\{\widetilde{k}-1, \, \widetilde{k}, \, \widetilde{k}+1 \})$ \hfill 
\\[4pt]
$(a, b, \_) = \texttt{leapfrog-orbit-selection}(\theta, \rho, B, h, 2^{k})$  \hfill (step size is $2^{-k} \, h$) \\ [4pt]
$(\theta^*, \rho^*, L) = \texttt{leapfrog-index-selection}(\theta, \rho, a, b, h, 2^{k})$   \\ [4pt]
$\widetilde{a} = a$, $\widetilde{b} = b$  \hfill (compute $B^*$) \\ [4pt]
 for $i$ from $\log_2(b-a+1)$ down to $1$:
\\[-12pt]
\null \qquad $m=\lfloor (\widetilde{a}+\widetilde{b})/2 \rfloor $ \\ [-12pt]
\null \qquad if $L \in [\widetilde{a}:m]$:
\\[-12pt]
\null \qquad \qquad $B_{i}^* = 0$ , $\widetilde{b} = m$
\\[-12pt]
\null \qquad else if $L \in [m+1:\widetilde{b}]$:
\\[-12pt]
\null \qquad \qquad $B_{i}^* = 1$, $\widetilde{a} = m+1$ 
\\[4pt]
 for $i$ from $\log_2(b-a+1)+1$ to $M$: 
\\[-12pt]
\null \qquad $B_{i}^* = B_{i}$ 
\\[4pt]
$\widetilde{k}^* = \texttt{step-reduction}(\theta^*, \rho^*, B^*, h, a_{\textrm{min}})$
\\[4pt]
if $k \in \{\widetilde{k}^*-1, \widetilde{k}^*, \widetilde{k}^*+1 \}$:  \hfill (GIST accept/reject step)
\\[-12pt]
\null \qquad  return $  \theta^*$ \hfill 
 \\[2pt]
 else: 
\\[-12pt]
 \null \qquad  return $  \theta$ \hfill 
 \\[4pt]
\vspace*{4pt}
\hrule
\caption{\textit{The no-U-turn sampler in the GIST framework with locally adaptive step size}.   }
\label{algo:adaptive-NUTS}
\end{flushleft}
\end{algorithm}

\begin{algorithm}[t]
\begin{flushleft}
$\texttt{step-reduction}(\theta, \rho, B, h, a_{\textrm{min}})$
\vspace*{2pt}
\hrule
\vspace*{2pt}
\textrm{Inputs:}
\begin{tabular}[t]{ll}
$\theta \in \mathbb{R}^{d}$ & position \\[2pt] 
$\rho \in \mathbb{R}^{d}$ & momentum \\[2pt] 
$B \in \{0, 1\}^M$ & Bernoulli directions \\[2pt]
$h>0$ &  coarse step size  \\[2pt] 
\end{tabular} 
\vspace*{4pt}
\hrule
\vspace*{8pt}
for $i$ from $1$ to $\infty$:  \hfill  (compute step size reduction factor from $(\theta, \rho)$)
\\[-12pt]
\null \qquad $(\_, \_,  \Delta H_{\text{gap}}) = \texttt{leapfrog-orbit-selection}(\theta, \rho, B, h, 2^{i})$ \\ [-12pt]
\null \qquad if $e^{-\Delta H_{\text{gap}}} \ge a_{\textrm{min}}:$ \hfill (if $\Delta H_{\text{gap}}$ within threshold) 
\\[-12pt]
\null \qquad \qquad  return $i$
\caption{\textit{Step reduction algorithm.}}\label{algo:step-reduction}
\end{flushleft}
\end{algorithm}


\subsection{Definition of \texorpdfstring{$B^*$}{BstarDef}.}

\label{sec:Bstar}

Unlike standard NUTS, ensuring the reversibility of step-size-adaptive NUTS requires constructing the Bernoulli process $B^*$ for each transition step.  This process is essential for determining the correct acceptance probability in \eqref{eq:step-size-acceptance-ratio}.  We now describe the procedure to construct this process given $B$, $[a:b]$, and $L$.  An implementation of this procedure is provided in Algorithm~\ref{algo:adaptive-NUTS}.   

We begin by defining a recursive function $(\beta_i(L, a, b))_{i \in [1:\ell]}$ that generates a binary string. If $a \ne b$, we define $\ell \in [1:M]$ such that $b-a+1 = 2^\ell$. This value of $\ell$ is well-defined because the random doubling procedure in NUTS results in an orbit of size $2^{\ell}$ for some $\ell \in [1:M]$.  In this case, the function $(\beta_i(L, a, b))_{i \in [1:\ell]}$ is defined as follows \[
\beta_i(L, a, b) = \begin{cases} \beta_i(L, a'(L), b'(L)) & \text{if  } \;i \in [1:\ell-1], \\
\mathbb{1}_{[m+1:b]}(L) & \text{if  } \;i = \ell,
\end{cases} \]
where $m = \lfloor (a+b)/2 \rfloor$ and 
\[
(a'(L), b'(L)) = 
\begin{cases} (a, m) & \text{if $L \in [a:m]$}, \\
(m +1 , b) & \text{if $L \in [m+1:b]$}.
\end{cases} 
\] 
If $a = b$,  $(\beta_i(L, a, b))_{i \in \emptyset}$ is defined as the empty binary string. Next, we define $B^*_i$ as follows:  
\begin{equation} 
B_i^*  = \begin{cases}
\beta_i(L, a,b) & \text{if } i \in [1 : \ell] \;, \\
B_i &  \text{if } i \in [\ell+1 : M] \;.
\end{cases}
\end{equation}
It can be shown by induction that for all $L \in [a : b]$ and $k \in \mathbb{N}$, the following equalities hold:  
\[
a^* = a - L \ \textrm{ and } \  b^* = b - L \;,
\]
\vspace*{-4pt}
where
\[
(a, b, \_) = \texttt{leapfrog-orbit-selection}(\theta, \rho, B, h, 2^k) \;,
\]
\[
 (a^*, b^*, \_) = \texttt{leapfrog-orbit-selection}(\theta^*, \rho^*, B^*, h, 2^k) \;,
\] 
and $(\theta^*, \rho^*) = \Phi_{2^{-k}h}^{2^kL}(\theta, \rho)$.

\section{Neal's funnel}
\label{sec:funnel}

\subsection{Target distribution}\label{sec:funnel-setup}

Consider the $(d+1)$-dimensional funnel distribution introduced by Neal \cite{Neal2003Slice}. This distribution is defined over variables $(\omega, x)$ where $\omega \in \mathbb{R}$ and $x = (x_1, \dots, x_{d}) \in \mathbb{R}^{d}$, with density 
\begin{equation} \label{eq:original-funnel}
p(\omega, x_1, \dots, x_d) = \text{normal}(\omega \mid 0, 3) \prod_{i=1}^d \text{normal}(x_i \mid 0, e^{\omega/2}) \;,
\end{equation}
where $\text{normal}(y \mid \mu, \sigma) = \frac{1}{ \sigma \sqrt{2 \pi}} e^{- \frac{1}{2} \left( \frac{y - \mu}{\sigma}\right)^2} $ is the density of a normal distribution with outcome $y \in \R$, location (i.e., mean) $\mu \in \R$, and scale (i.e., standard deviation) $\sigma \in (0, \infty)$. 


\subsection{Pathology with fixed step size}

The funnel distribution serves as a minimal example demonstrating the pathologies found in Bayesian hierarchical models \cite{Neal2003Slice, betancourt2013hamiltonianmontecarlohierarchical}.  
 As discussed in these references and as we will briefly review, this distribution exhibits multiscale behavior typical of many hierarchical Bayesian models, which often hinders the mixing properties of samplers that use a fixed step size, including NUTS.

In the region $\{ \omega > 0 \}$, the potential energy function $U(\omega, x_1, \dots, x_d)  = -\log(p(\omega, x_1, \dots, x_d))$ has low curvature. Efficient exploration of this region using NUTS requires a relatively large leapfrog step size.  Conversely, in the region $\{ \omega < 0 \}$, the potential energy has high curvature, necessitating a very small step size for the leapfrog integrator used by NUTS to remain stable.  Specifically, for the leapfrog integrator to be stable, the step size $h$ must satisfy $\sqrt{L}h \leq 2$, where $L$ is the local Lipschitz constant of $(D U)$ in a neighborhood of the leapfrog path.  This step size restriction is due to the fact that the leapfrog integrator is only conditionally stable \cite{LeRe2004}.

This local Lipschitz constant $L$ can be upper bounded by the spectral radius of the Hessian matrix $D^2 U$.  In Figure~\ref{fig:cutoff-log-sigma-h-.25}, we use this upper bound to estimate the region of the funnel distribution where the leapfrog integrator becomes unstable. Along the $\omega$-axis,  \[
D^2 U(\omega, 0, \dots, 0) = \operatorname{diag}(1/9, e^{-\omega}, \dots, e^{-\omega}) 
\] and hence, its spectral radius is equal to $\max(1/9, e^{-\omega})$. Using this estimate in the leapfrog stability condition, we find that stability requires $e^{-\omega/2} h \le 2$. Therefore, to explore the region $\omega \in (- \beta, \infty)$, the step size must be no larger than $h = 2 e^{-\beta/2}$. This means that exploring the region $\{ \omega <0 \} $ with sufficiently high precision requires exponentially smaller step sizes. 

As a result, no fixed-step-size NUTS algorithm can stably and efficiently explore both the $\{ \omega > 0 \}$ and the $\{ \omega < 0 \}$ regions of the funnel distribution. Since each of these regions contains half the total mass of the distribution, the difficulties posed by these regions will likely be encountered during sampling, and if the leapfrog step size is fixed, the smallest necessary step size must be used throughout state space, which can impair mixing in terms of computational cost.  Therefore, this funnel distribution strongly motivates and rigorously tests our step-size-adaptive NUTS algorithm.  

\subsection{Results with step size adaptation}

\begin{figure}[t]
    \centering
    \begin{subfigure}[b]{0.46\textwidth}
        \includegraphics[width=\textwidth]{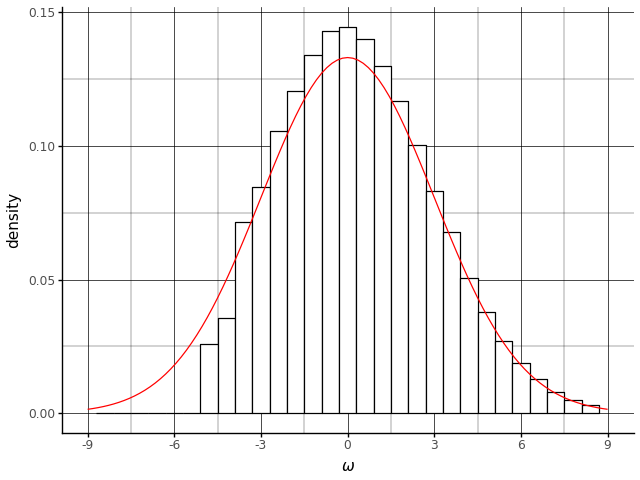}
        \caption{Histogram of $\omega$ for NUTS at $h=0.25$}
        \label{fig:log-sigma-marginal-histogram}
    \end{subfigure}
    \hfill 
    \begin{subfigure}[b]{0.46\textwidth}
        \includegraphics[width=\textwidth]{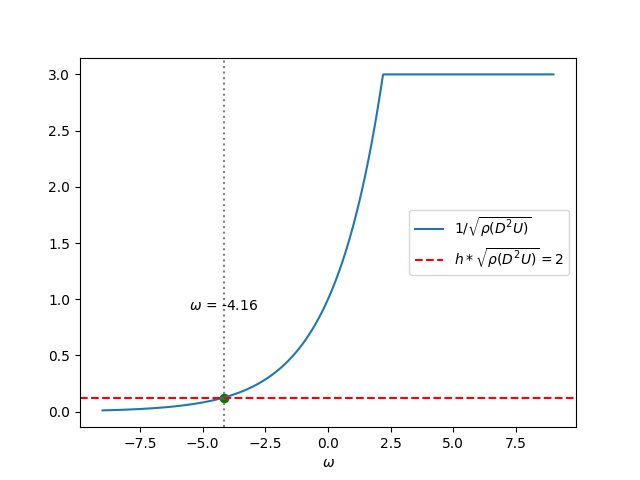}
        \caption{Stability region of leapfrog for $h=.25$}
        \label{fig:cutoff-log-sigma-h-.25}
    \end{subfigure}
    %
    \begin{subfigure}[b]{0.46\textwidth}
        \includegraphics[width=\textwidth]{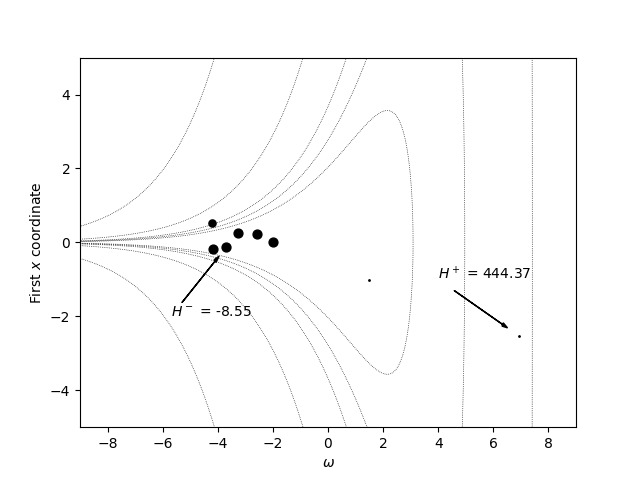}
        \caption{NUTS path with initial $\omega = -2.0$}
        \label{fig:orbit-initial-log-sig-minus2}
    \end{subfigure}
    \hfill
    \begin{subfigure}[b]{0.46\textwidth}
        \includegraphics[width=\textwidth]{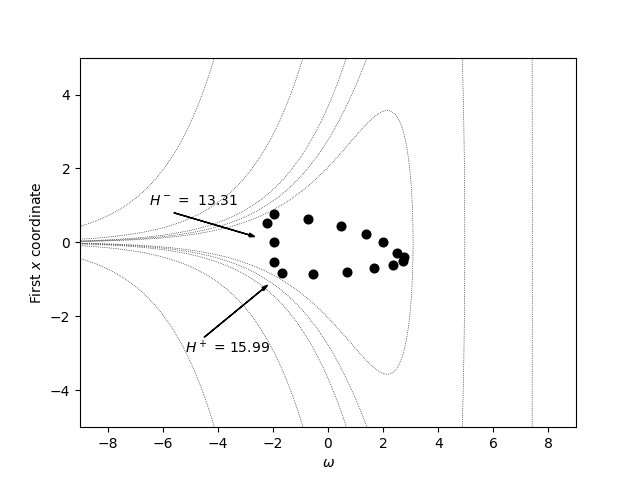}
        \caption{NUTS path with initial $\omega = 2.0$}
        \label{fig:orbit-initial-log-sig-2}
    \end{subfigure}

    \vspace{0.6cm}
    
    \caption{\textit{For NUTS with fixed step size $h=1/4$: (a) Histogram of $\omega$ from 250,000 draws, showing a bottleneck at around $\omega = -4$. (b) Predicted bottleneck based on the stability argument in Section \ref{sec:funnel-setup}, matching the observed bottleneck.  (c) NUTS path initialized at $(-2, 0,\dots,0)$ with initial velocity $(-2, 1,\dots,1)$ rescaled to have norm $\sqrt{11}$, showing instability near $\omega=-4$. (d) NUTS path initialized at $(2, 0,\dots,0)$ with the same initial velocity, showing stability throughout the path. 
     In (c) and (d), the initial velocity is rescaled to have a norm of $\sqrt{11}$, typical for an $11$-dimensional standard normal random variable.  The size of the dots is proportional to $-H(\theta, \rho)$, with smaller dots representing higher energy points and lower Boltzmann weight.}}
    \label{fig:cutoff-value-for-log-sigma}
\end{figure}

Figure~\ref{fig:log-sigma-marginal-histogram} shows a histogram of 250,000 samples produced by NUTS with fixed step size $h= 1/4$ applied to the funnel distribution described above. These samples exhibit a sharp drop near $\omega = -4$ with essentially no samples for values of $\omega$ less than $-4$.  This phenomenon is due to the instability of the leapfrog integrator in this region, as depicted in Figure~\ref{fig:cutoff-log-sigma-h-.25}.

In contrast, Figure~\ref{fig:metropolizable-adaptive-nuts} shows a histogram for 250,000 samples generated by our step-size-adaptive NUTS algorithm given in Algorithm \ref{algo:adaptive-NUTS} with $h = 1/2$, $M=10$, and $a_{min} = 0.7$. Unlike the samples produced by NUTS with a fixed step size, these samples do not show any drop-off. Our step-size-adaptive NUTS algorithm effectively generates samples even for large negative values of $\omega$, demonstrating its effectiveness in this challenging stiff scenario.  

\begin{figure}
\centering
\begin{subfigure}{0.45\textwidth} 
\includegraphics[scale=0.45]{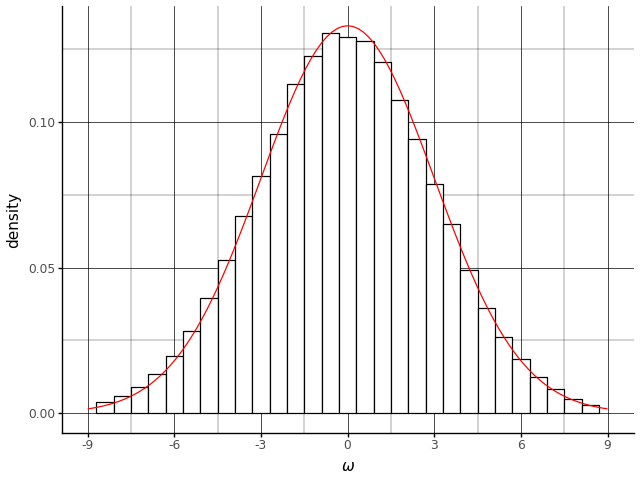} 
\caption{\textit{Histogram of $\omega$ for Algorithm \ref{algo:adaptive-NUTS}}}
\label{fig:metropolizable-adaptive-nuts-log-sigma}
\end{subfigure}
\begin{subfigure}{0.45\textwidth} 
\includegraphics[scale=0.45]{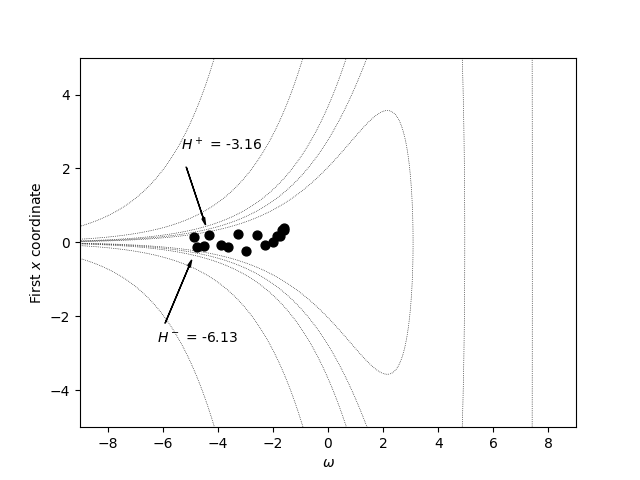} 
\caption{\textit{Example intermediate orbit in Algorithm \ref{algo:adaptive-NUTS}}}
\label{fig:example-orbit} 
\end{subfigure}
\caption{\textit{(a) Histogram of 250,000 samples generated by Algorithm \ref{algo:adaptive-NUTS} for the funnel distribution defined by \eqref{eq:original-funnel} with $h = 1/2$.  (b) Example orbit from the step size selection search, initialized like the one in Figure~\ref{fig:orbit-initial-log-sig-minus2}, but produced using $\Phi_{h/2}^2$ instead of $\Phi_{h}$ for $h=1/4$, showing greater stability compared to the leapfrog path shown in Figure~\ref{fig:orbit-initial-log-sig-minus2}.}}
\label{fig:metropolizable-adaptive-nuts}
\end{figure}

\section{High-dimensional standard normal}

\label{sec:highDnormal}

Consider a $d$-dimensional standard normal target distribution defined over the variables $x = (x_1, \dots, x_{d}) \in \mathbb{R}^{d}$, with density 
\begin{equation} \label{eq:highDnormal}
p(x_1, \dots, x_d) = \prod_{i=1}^d \text{normal}(x_i \mid 0, 1) \;.
\end{equation}
In this scenario, energy errors can accumulate with increasing dimension $d$, creating bottlenecks even when the leapfrog step size meets the stability requirement $h \in (0,2)$.  However, due to a concentration phenomenon, these energy errors typically decrease when the chain enters the typical set where the target distribution concentrates.   Specifically, outside the typical set, the leapfrog energy error scales as $O(h^2 d)$.  Remarkably, within the typical set, this scaling improves to $O(h^4 d)$ \cite{BePiRoSaSt2013}.   As illustrated in Figure~\ref{fig:highDnormal}, this difference in scaling is seamlessly detected by step-size-adaptive NUTS, demonstrating its effectiveness in this challenging high-dimensional scenario.

More precisely, the energy error after $L$ leapfrog steps is given by $$\Delta H(\theta, \rho) \ =\ \frac{1}{2} |\rho^*|^2 + U(\theta^*) - \frac{1}{2} |\rho|^2 - U(\theta) \ =
\ \frac{h^2}{8}\left( |\theta^*|^2-|\theta|^2\right)\,$$ 
where $(\theta^*,\rho^*)\ = \Phi_h^L(\theta,\rho)$.  From the mode, $\Delta H = \mathcal{O}(h^2d)$, implying that $h=\mathcal{O}(d^{-1/2})$ is required for non-degenerate acceptance in the leapfrog-index-selection step. However, if $\theta$ starts from stationarity,  then $|\theta|^2$ and $|\theta^*|^2$ concentrate around $d$ with fluctuations of order $\mathcal O(d^{1/2})$. This leads to $\Delta H = \mathcal{O}(h^2d^{1/2})$, suggesting that $h=\mathcal{O}(d^{-1/4})$ is more appropriate in this scenario.


\begin{figure}
\centering
\begin{subfigure}{0.45\textwidth} 
\includegraphics[scale=0.45]{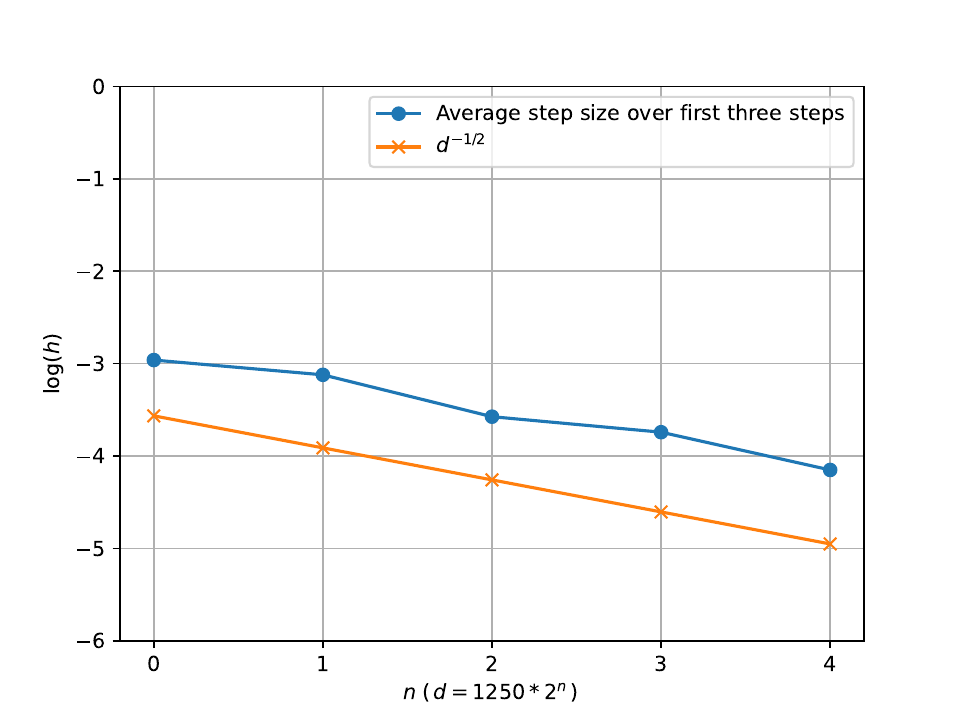} 
\caption{\textit{away from the high-density region}}
\label{fig:dpt5}
\end{subfigure}
\begin{subfigure}{0.45\textwidth} 
\includegraphics[scale=0.45]{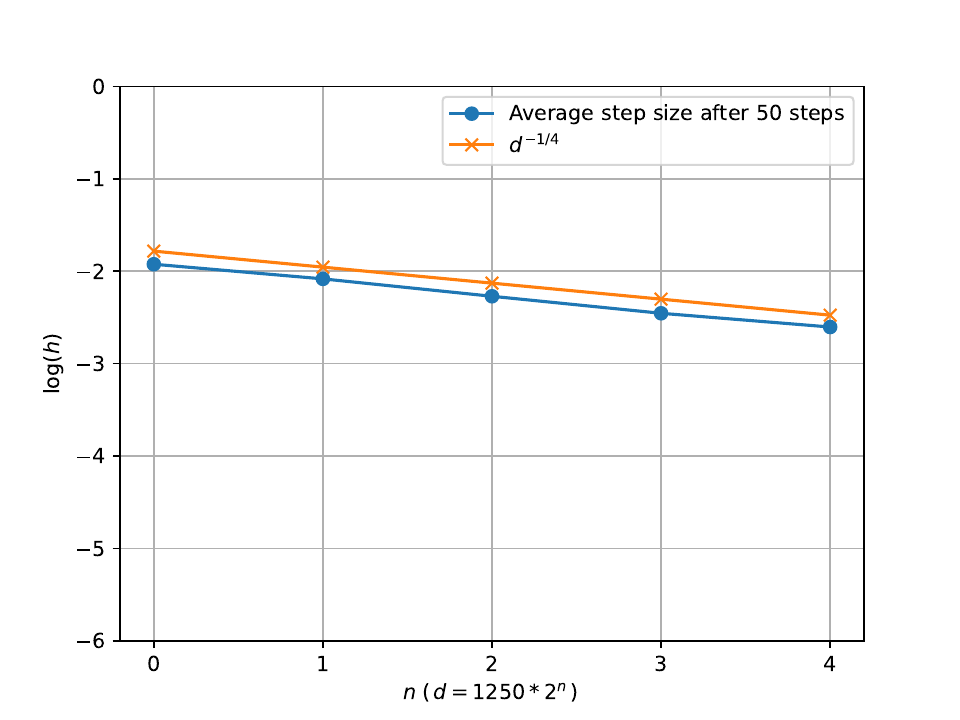} 
\caption{\textit{within the high-density region}}
\label{fig:dpt25} 
\end{subfigure}
\caption{\textit{(a) This plot shows the  average step size over the first three steps of 800 chains produced by Algorithm \ref{algo:adaptive-NUTS} with $h = 1/2$, $M=10$, and $a_{min} = 0.7$ each starting at the mode of the standard Gaussian distribution defined in \eqref{eq:original-funnel} with $h = 1/2$.  In this initial regime, the adaptive step size used by Algorithm \ref{algo:adaptive-NUTS} scales as $O(d^{-1/2})$. (b)  This plot shows the average step size after 50 steps, by which time the chains have  reached the high-density region of the target distribution. In this regime, the adaptive step size scales as $O(d^{-1/4})$.}}
\label{fig:highDnormal}
\end{figure}

\appendix 

\section{Reversibility of step-size-adaptive NUTS} \label{app:proof}

We have defined the conditional probability for step-size-adaptive NUTS in \eqref{eq:adapt-nuts-conditional} and  the transformation $G$ in \eqref{eq:adapt-nuts-involution}. To complete the verification that step-size-adaptive NUTS is an instance of GIST, we now need to show that the transformation $G$ is a measure-preserving involution and confirm that the acceptance ratio in \eqref{eq:step-size-acceptance-ratio} matches  the corresponding GIST acceptance probability.  

\subsection{Notation}
Let $B$ be a binary string of length $M$. Denote by $(B_i)_{i \in [1:\ell]}$ the first $\ell$ bits of $B$. 
Define functions $a_\ell((B_i)_{i \in [1:\ell]})$ and $b_\ell((B_i)_{i \in [1:\ell]})$ inductively as follows:
\begin{align}
a_0 &= 0, \\
b_0 &= 0 , \\
\tilde{a}_\ell((B_i)_{i \in [1:\ell]}) &= a_{\ell-1}((B_i)_{i \in [1:\ell-1]}) + (-1)^{B_\ell} \cdot 2^{\ell-1} \, , \\
\tilde{b}_\ell((B_i)_{i \in [1:\ell]}) &= b_{\ell-1}((B_i)_{i \in [1:\ell-1]}) + (-1)^{B_\ell} \cdot 2^{\ell-1}\, , \\
a_\ell((B_i)_{i \in [1:\ell]}) &= \min(a_{\ell-1}((B_i)_{i \in [1:\ell-1]}), \tilde{a}_\ell((B_i)_{i \in [1:\ell]}))\, , \label{eq:inductive-def-of-a-j} \\
b_\ell((B_i)_{i \in [1:\ell]}) &= \max(b_{\ell-1}((B_i)_{i \in [1:\ell-1]}), \tilde{b}_\ell((B_i)_{i \in [1:\ell]}))\,. \label{eq:inductive-def-of-b-j}
\end{align}
Note that 
\[b_\ell((B_i)_{i \in [1: \ell]}) - a_{\ell}((B_i)_{i \in [1:\ell]}) + 1 = 2^{\ell}\]
and 
\begin{multline*}    
[a_{\ell+1}((B_i)_{i \in [1:\ell+1]}) : b_{\ell+ 1}((B_i)_{i \in [1:\ell+1]}) ]  \\= [a_{\ell}((B_i)_{i \in [1:\ell]}) : b_{\ell}((B_i)_{i \in [1:\ell]}) ] \cup [\tilde{a}_{\ell+1}((B_i)_{i \in [1:\ell+1]}) : \tilde{b}_{\ell+ 1}((B_i)_{i \in [1:\ell+1]}) ] \, , 
\end{multline*}
with one set being the top half and the other set being the bottom half. 

Let \begin{equation}
\begin{aligned}
& (a_T, b_T, \_) = (a_T(\theta, \rho, B, h, R), b_T(\theta, \rho, B, h, R), \_) \\
& \qquad  = \texttt{leapfrog-orbit-selection}(\theta, \rho, B, h, R).
\end{aligned}
\end{equation}
From the definition of \texttt{leapfrog-orbit-selection}, we obtain
\begin{align*}
a_T(\theta, \rho, B, h, R) &= a_{\ell_{min}} ((B_i)_{i \in [1:\ell_{min}]}) \, , \\
b_T(\theta, \rho, B, h, R) &= b_{\ell_{min}}((B_i)_{i \in [1:\ell_{min}]}) \, ,
\end{align*}
where $\ell_{min} = \ell_{min}(\theta, \rho, (B_i)_{i \in [1 : M]}, h, R)$ is defined by 
\begin{align*}
& \ell_{min}( \theta, \rho, B, h, R)
\\
&  \quad = \min( \big \{\ell \in [1:M] \mid \mathbb{1}_{\text{U-turn}}(a_\ell((B_i)_{i \in [1:\ell]}), b_\ell((B_i)_{i \in [1:\ell]}), \theta, \rho, h, R) = 1  \\ & \hspace{1 cm} \text{  or     } \hspace{.5 cm} \mathbb{1}_{\text{sub-U-turn}}(\tilde{a}_{\ell+1}((B_i)_{i \in [1:\ell+1]}), \tilde{b}_{\ell+1} ((B_i)_{i \in [1:\ell+1]}), \theta, \rho, h, R) = 1 \big \}) \wedge M 
\end{align*}
with $\mathbb{1}_{\text{U-turn}}$ and $\mathbb{1}_{\text{sub-U-turn}}$ defined as in \eqref{eq:uturn} and \eqref{eq:subuturn}, and $\min(\emptyset) = \infty$. 

In the proof of correctness we will make use of following shift invariance of the U-turn and sub-U-turn conditions:
\begin{align} \label{eq:u-turn-symmetry}
\mathbb{1}_{\text{U-turn}}(\Phi_{h/R}^{RL}(\theta, \rho), a -L, b-L, h, R) &= \mathbb{1}_{\text{U-turn}}(\theta, \rho, a, b, h, R) \, , \\
\mathbb{1}_{\text{sub-U-turn}}(\Phi_{h/R}^{RL}(\theta, \rho), a -L, b-L, h, R) &= \mathbb{1}_{\text{sub-U-turn}}(\theta, \rho, a, b, h, R)\,.
\end{align}
These follow immediately from the definition of the U-turn set in \eqref{eq:uturn}.

The functions $a_\ell$ and $b_\ell$ admit explicit expressions in terms of $(B_i)_{i \in [1:\ell]}$. Indeed, one observes 
\begin{align}
a_\ell( (B_i)_{i \in [1:\ell]}) &= a_{\ell-1}((B_i)_{i \in [1:\ell-1]}) - 2^{\ell-1} \mathbb{1}_{\{B_\ell =1\}}  = a_{\ell-1}((B_i)_{i \in [1:\ell-1]}) -  2^{\ell-1} B_\ell \, , \\
b_\ell((B_i)_{i \in [1:\ell]}) &= b_{\ell-1}((B_i)_{i \in [1: \ell-1]}) + 2^{\ell-1} \mathbb{1}_{\{B_\ell =0\}}  = b_{\ell-1}((B_i)_{i \in [1:\ell-1]}) + 2^{\ell-1}  (1-B_\ell) \, .
\end{align}

Together with the condition $a_0 = b_0 = 0$ the above yields
\begin{align}
a_\ell((B_i)_{i \in [1:\ell]}) = -\sum_{j=1}^\ell 2^{j-1}B_j \, , \label{eq:explicit-def-of-a-j}\\
b_\ell((B_i)_{i \in [1:\ell]}) = \sum_{j=1}^\ell 2^{j-1}(1-B_j) \label{eq:explicit-def-of-b-j} \, .
\end{align}

We recall the definitions of $(\beta_i(L, a, b))_{i \in [1:\ell]}$ and $B^*$ appearing in Section \ref{sec:adaptNUTS}. 
To review, for any $L,a,b \in \mathbb{Z}$ such that $L \in [a:b]$ and $b-a+1 = 2^{\ell}$ the function $(\beta_i(L, a, b))_{i \in [1:\ell]}$ is defined recursively by
\[
\beta_i(L, a, b) = \begin{cases} \beta_i(L, a'(L), b'(L)) & \text{if  } \;i \in [1:\ell-1] \\
\mathbb{1}_{[m+1:b]}(L) & \text{otherwise.}
\end{cases} 
\]
where $(\beta_i(L,a,b))_{i \in \emptyset}$ is defined as the empty string to serve as a base case. Define $B^*$ as 
\begin{equation} 
B_i^* = B_i^*(L,a,b,B) = \begin{cases}
\beta_i(L, a,b) & i \in [1 : \ell] \\
B_i & i \in [\ell+1 : M]
\end{cases}
\end{equation}
where above we have emphasized the dependence of $B^*$ on $L,a,b,B$. For brevity we suppress this explicit dependence whenever possible. Figure \ref{fig:example-beta} shows an example of the computation of $(\beta_i(L, a, b))_{i \in [1:\ell]}$.

In the above notation, we can write the distribution $P(a, b, \ell \mid \theta, \rho, B, h, R)$ appearing in \eqref{eq:NUTS-kernel} as 
\begin{equation} \label{eq:orbit-selection-kernel-explicit}
P(a,b, \ell \mid \theta, \rho, B, h, R)  = \delta_{\ell_{min}(\theta, \rho, B, h, R)}(\ell) \; \delta_{a_{\ell}((B_i)_{i \in [1:\ell]})}(a) \; \delta_{b_{\ell}((B_i)_{i \in [1:\ell]})}(b).
\end{equation}

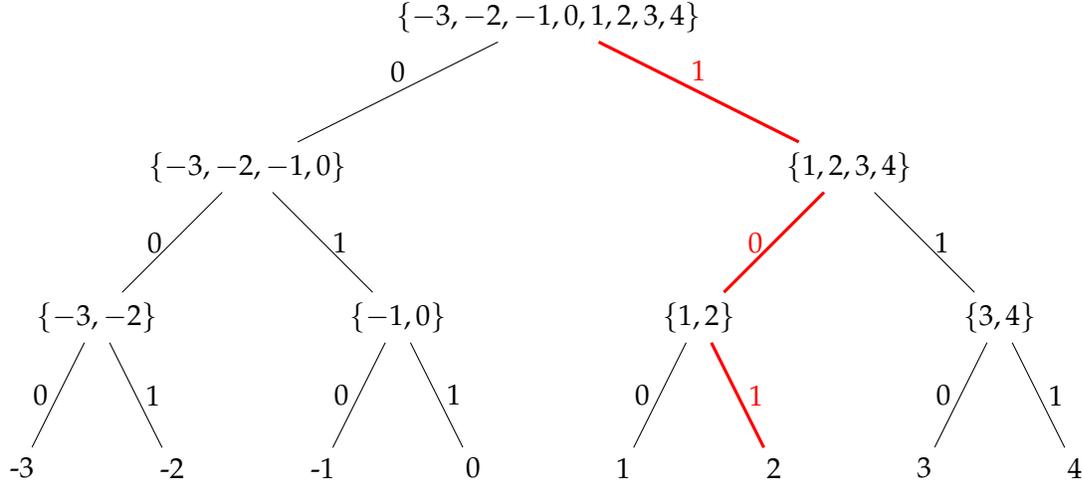
\begin{figure} 
\begin{tikzpicture} [
    level 1/.style={sibling distance = 8cm, level distance = 2cm},
    level 2/.style={sibling distance = 4cm},
    level 3/.style={sibling distance = 2cm},
    emph/.style={edge from parent/.style={red, very thick, draw}},
    norm/.style={edge from parent/.style={black, thin, draw}}
    ]

    \node{$\{-3,-2,-1,0,1,2,3,4\}$}
    child[norm] { node {$\{-3,-2,-1,0\}$}
        child[norm] { node {$\{-3,-2\}$}
            child { node {-3} edge from parent node[midway, left] {0} }
            child { node {-2} edge from parent node[midway, right] {1} }
            edge from parent node[midway, left]{0}
        }
        child[norm] { node {$\{-1,0\}$}
            child { node {-1} edge from parent node[midway, left] {0} }
            child { node {0} edge from parent node[midway, right] {1} }
            edge from parent node[midway, right]{1}
        }
        edge from parent node[midway, above]{0}
    }
    child[emph] { node {$\{1,2,3,4\}$}
        child[emph] { node {$\{1,2\}$}
            child[norm] { node {1} edge from parent node[midway, left] {0} }
            child[emph] { node {2} edge from parent node[midway, right] {1} }
            edge from parent node[midway, left]{0}
        }
        child[norm] { node {$\{3,4\}$}
            child { node {3} edge from parent node[midway, left] {0} }
            child { node {4} edge from parent node[midway, right] {1} }
            edge from parent node[midway, right]{1}
        }
        edge from parent node[midway, above] {1}
    };

\end{tikzpicture}
\caption{Computation of $(\beta_i(2,-3,4))_{i \in [1:3]}$. Here, $(\beta_i(2,-3,4))_{i \in [1:3]} = (1,0,1)$. }
\label{fig:example-beta}
\end{figure}

\subsection{Proof of reversibility}
\begin{theorem}
The transition kernel for the Markov Chain generated by Algorithm \ref{algo:adaptive-NUTS} is reversible with respect to the target distribution.
\end{theorem}

\begin{proof}
By Theorem 1 in \cite{BouRabeeCarpenterMarsden2024}, we need only establish the function 
\begin{equation} 
G: (\theta, \rho, k, B, \ell, a, b, L) \mapsto (\Phi_{2^{-k} h}^{2^k L}(\theta, \rho), k, B^*,\ell, a-L, b-L, -L) \;.
\end{equation}
appearing in \eqref{eq:adapt-nuts-involution} is a measure-preserving involution and verify (\ref{eq:step-size-acceptance-ratio}) is the corresponding GIST-acceptance probability. 

$G$ is not defined on all of $\mathbb{R}^{2d} \times \mathbb{N} \times \{0, 1\}^M \times [1:M] \times \mathbb{Z}^3$, but rather is defined on the smaller domain $\mathbb{R}^{2d} \times \mathbb{N} \times \mathcal{D}$ where 
\[\mathcal{D} = \{ (B, \ell, a, b, L) \in \{0,1\}^M \times [1:M] \times \mathbb{Z} \mid a = a_{\ell}((B_i)_{i \in [1:\ell]}), \; b = b_\ell((B_i)_{i \in [1: \ell]}),\; a \leq L \leq b \}.\]
This domain contains the support of the joint distribution for Algorithm \ref{algo:adaptive-NUTS}. Consequently, the resulting Markov Chain will always take values in this restricted set and this restriction has no impact on the definition of Algorithm \ref{algo:adaptive-NUTS}. 

We show $G$ is a measure-preserving involution in Lemma \ref{lem:G-measure-preserving}. 
The acceptance probability  (\ref{eq:step-size-acceptance-ratio}) follows immediately from Corollary \ref{cor:NUTS-kernel-reversibility} of Lemma \ref{lem:orbit-selection-symmetry}. Combining these facts implies the transition kernel of the Markov Chain generated by Algorithm \ref{algo:adaptive-NUTS} is reversible by Theorem 1 in \cite{BouRabeeCarpenterMarsden2024}.
\end{proof}

\begin{lemma} \label{lem:G-measure-preserving}
The map $G: \mathbb{R}^{2d} \times \mathbb{Z} \times \mathcal{D} \to \mathbb{R}^{2d} \times \mathbb{Z} \times \mathcal{D}$ is a measure-preserving involution
\end{lemma}

\begin{proof}
Define \[\Gamma(B, \ell, a, b, L) = (B^*, \ell, a-L, b -L, -L).\]
In this notation, we can write 
\[G(\theta, \rho, k, B, \ell, a, b, L) = (\Phi_{2^{-k} h}^{2^k L} (\theta, \rho), k, \Gamma(B, \ell, a, b, L)).\]
We first check that $\mathcal{D}$ is in fact the codomain of the map $\Gamma$ as defined above. Using Lemma \ref{lem:initial-point-symmetry} and the definition of $B^*$ we obtain
\begin{align}
a_\ell((B^*_i)_{i \in [1:\ell]}) &=  a_\ell((\beta_i(L, a, b))_{i \in [1:\ell]}) \notag \\ 
&= a-L \, ,\label{eq:a-ell-in-D} \\
b_\ell((B^*_i)_{i \in [1:\ell]}) & = b_\ell((\beta_i(L, a, b))_{i \in [1:\ell]}) \notag \\ 
&= b-L \, .\label{eq:b-ell-in-D}
\end{align}
Additionally, $a-L \leq -L \leq b- L $ since $a \leq 0 \leq b$. Combining these, $\Gamma(B, \ell, a, b, L) \in \mathcal{D}$. This additionally implies that $G$ is defined as a map
$\mathbb{R}^{2d} \times \mathbb{Z} \times \mathcal{D} \to \mathbb{R}^{2d} \times \mathbb{Z} \times \mathcal{D}$. 

We next show that $G$ is an involution. First, we show $\Gamma$ is an involution.
\begin{align}
& \Gamma \circ \Gamma (B, \ell, a, b, L) &  \notag\\
& \quad =\Gamma((B_i^*(L,a,b,B))_{i \in [1:M]}, \ell, a-L, b-L, -L) \notag \\
& \quad = ((B_i^*(-L, a-L, b-L, (B_i^*(L,a,b,B))_{i \in [1:M]}))_{i \in [1:M]}, \ell, a, b, L).\notag 
\end{align}
To establish that $\Gamma$ is an involution we need only show  
\begin{equation} \label{eq:b-star-involutivity}
\forall \ell' \in [1:M], \quad B_{\ell'}^*(L, -(b-L), -(a-L), (B_i^*(L,a,b,B))_{i \in [1:M]})= B_{\ell'}.
\end{equation}
Since $(B, \ell, a, b, L) \in \mathcal{D}$ we know $(b-L)-(a-L) + 1 = b -a +1 = 2^\ell$. Therefore, by definition for $\ell' \in [\ell +1 : M] $
\begin{align} 
B_{\ell'}^*(-L, a-L, b-L, (B_i^*(L,a,b,B))_{i \in [1:M]}) &= B_{\ell'}^*(L,a,b,B) = B_{\ell'} .\label{eq:b-start-involutity-direct-part}
\end{align}
On the other hand, using Lemma \ref{lem:initial-point-symmetry} 
\begin{align}
& a_\ell( (B_i^* (-L, a-L, b-L,  (B_i^*(L,a,b,B))_{i \in [1:M]}))_{i \in [1: \ell]}) & \notag \\
& \quad = a_\ell((\beta_i(-L, a-L, b-L))_{i \in [1:\ell]})  = (a - L) -(-L) = a \;. \notag 
\end{align}
Thus, using (\ref{eq:explicit-def-of-a-j}) we have
\[\sum_{j = 1}^\ell 2^{j-1}B_j^*(-L, a-L, b-L, (B_i^*(L,a,b,B))_{i \in [1:M]})  = \sum_{j=1}^\ell 2^{j-1} B_j.\]
Therefore,
\begin{equation} \label{eq:b-star-involutivity-binary-part}
\forall \ell' \in [1:\ell], \quad B_{\ell'}^*(-L, a-L, b-L, (B_i^*(L,a,b,B))_{i \in [1:M]}) = B_{\ell'} 
\end{equation}
by the uniqueness of binary expansions (see Remark \ref{rmk:binary-expansions} for a proof sketch). Combining \eqref{eq:b-start-involutity-direct-part} and \eqref{eq:b-star-involutivity-binary-part} gives (\ref{eq:b-star-involutivity}). 

Consequently, $\Gamma$ is an involution. To show that $G$ is an involution write
$G(\theta, \rho, k, B, \ell, a, b, L) = (\Phi_{2^{-k} h}^{2^k L}(\theta, \rho), k, \Gamma(B, \ell, a, b, L))$
and compute
\begin{align}
G \circ G (\theta, \rho, k, B, \ell, a, b, L) &= 
G(\Phi_{2^{-k}h}^{2^k L}(\theta, \rho),k, \Gamma (B, \ell, a, b, L)) \notag \\
&= (\Phi_{2^{-k}h}^{-2^k L}(\Phi_{2^{-k}h}^{2^k L}(\theta, \rho)), k, \Gamma \circ \Gamma (B, \ell, a, b, L)) \notag \\
&= (\theta, \rho, k, B, \ell, a, b, L). \notag
\end{align}

Last, we argue that $G$ is measure-preserving. For fixed $L,k$ the map $(\theta, \rho) \to \Phi_{2^{-k}h}^{2^k L}(\theta, \rho)$ is measure-preserving. As $\Gamma$ is a bijection it preserves the counting measure on $\mathcal{D}$. Combining these observations with an application of Fubini's theorem establishes that $G$ is measure-preserving.

\end{proof}

\begin{lemma} \label{lem:orbit-selection-symmetry}
If $(\theta, \rho) \in \mathbb{R}^{2d}$, $(B, \ell, a, b, L) \in \mathcal{D}$
then
\begin{equation} \label{eq:orbit-selection-detailed-balance}
P(a,b, \ell \mid \theta, \rho, B, h, R) = P(a-L, b-L, \ell \mid \Phi_{h/R}^{RL}(\theta, \rho), B^*,h, R) . 
\end{equation}
\end{lemma}

\begin{proof}
Define \[\Tilde{G}(\theta, \rho, B, \ell, a, b, L) = (\Phi_{h/R}^{RL}(\theta, \rho), \Gamma(B, \ell, a, b, L)),\] 
for $\Gamma$ defined in the proof of Lemma \ref{lem:G-measure-preserving}. As before, $\Tilde{G}$ is an involution. 

Let \[F(\theta, \rho, B, \ell, a, b, L) = P(a,b,\ell \mid \theta, \rho, B, h ,R).\] Note that $P$ is $\{0,1\}$-valued and hence so is $F$. We want to show \[F(\theta, \rho, B, \ell, a, b, L) = F(\Tilde{G}(\theta, \rho, B, \ell, a, b, L)).\]
As $F$ is $\{0, 1\}$-valued, it is enough to show
\begin{equation} \label{eq:F=1-iff-Fcircg=1}
F(\theta, \rho, B, \ell, a, b, L) = 1 \iff F(\Tilde{G}(\theta, \rho, B, \ell, a, b, L)) = 1 .
\end{equation}
To that end, suppose $F(\theta, \rho, B, \ell, a, b, L) = 1$. Then, 
\[a = a_T(\theta, \rho,B, h, R) \quad \text{and} \quad b = b_T(\theta, \rho, B, h, R).\]
Hence, by Lemma \ref{lem:j-min-symmetry}
\[\delta_{\ell_{min}(\theta, \rho, B, h, R)}(\ell) = 
\delta_{\ell_{min}(\Phi_{h/R}^{RL}(\theta, \rho), B^*, h, R)}(\ell) = 1.\]
Combining this with Lemma \ref{lem:initial-point-symmetry}, we find $F(\Tilde{G}(\theta, \rho, B, \ell, a, b, L)) = 1$. Suppose on the other hand that $F(\Tilde{G}(\theta, \rho, B, \ell, a, b, L)) = 1$. Then, by the above argument and the fact $\Tilde{G}$ is an involution 
\[F(\theta, \rho, B, \ell, a, b, L) = F(\Tilde{G} \circ \Tilde{G} (\theta, \rho, B, \ell, a, b, L))  = 1,\]
from which we deduce \eqref{eq:F=1-iff-Fcircg=1} and therefore \eqref{eq:orbit-selection-detailed-balance}.

\end{proof}

\begin{corollary} \label{cor:NUTS-kernel-reversibility} 
Suppose $Q$ satisfies
\[ e^{-H(\theta, \rho)} Q(L \mid \theta, \rho, a, b, h, R) = e^{-H \circ \Phi_{h/R}^{RL}(\theta, \rho)} Q(-L \mid \Phi_{h/R}^{RL}(\theta, \rho), a-L, b-L, h, R).\]
If $(\theta, \rho) \in \mathbb{R}^{2d}$ and $(B, \ell, a, b, L) \in \mathcal{D}$, then
\begin{align}
\label{eq:NUTS-detailed-balance}
e^{- H(\theta, \rho)} p_{\text{NUTS}}(B, \ell, a, b, L \mid \theta, \rho, h, R) = 
e^{-H \circ \Phi_{h/R}^{RL}(\theta, \rho)} p_{\text{NUTS}}(B^*, \ell, a-L, b-L, -L) .
\end{align}
\end{corollary} 

\begin{proof}
From Lemma \ref{lem:orbit-selection-symmetry} and the assumption on $Q$, \eqref{eq:NUTS-detailed-balance} immediately follows.

\end{proof}

\begin{remark}(Sketch of binary expansion uniqueness) \label{rmk:binary-expansions}

Suppose $(r_i)_{i \in [1:\ell]}, (q_i)_{i \in [1:\ell]}$ are binary strings and
\[\sum_{j=1}^\ell r_j 2^{j-1} = \sum_{j=1}^\ell q_j 2^{j-1}. \]
Reducing modulo 2 immediately gives $r_1 = q_1$. Subtracting this common value from both sides and dividing by 2 then gives 
\[\sum_{j=1}^{\ell -1} r_{j+1} 2^{j-1} = \sum_{j=1}^{\ell -1} q_{j+1} 2^{j-1} \] 
and repeating this process shows $(r_i)_{i \in [1:\ell]} = (q_i)_{i \in [1:\ell]}$.
\end{remark}

\subsection{Lemmas related to \texorpdfstring{$B^*$}{Bstar}}
\begin{lemma} \label{lem:initial-point-symmetry}
For $L \in [a:b]$, $b-a+1 = 2^\ell$

\begin{equation}
a_\ell((\beta_i(L, a, b))_{i \in [1:\ell]}) = a - L  \quad \text{and} \quad b_\ell((\beta_i(L, a, b))_{i \in [1:\ell]}) = b - L .
\end{equation}
\end{lemma}

\begin{proof}
It is slightly cleaner to equivalently prove,
\begin{equation*}
a- a_\ell((\beta_i(L, a, b))_{i \in [1:\ell]}) = L  \quad \text{and} \quad b- b_\ell((\beta_i(L, a, b))_{i \in [1:\ell]}) = L.
\end{equation*}

We proceed by induction on $\ell$. For $\ell = 0$, $b-a+1 = 1$. Thus, $b=a =L$ and we obtain directly
\begin{equation}
a - a_0 = a - 0 = L \quad \text{and} \quad b - b_0 = b -0 = L.
\end{equation}

For $\ell>0$, write
\begin{align}
a - a_\ell((\beta_i(L,a,b))_{i \in [1:\ell]}) &= a - (a_{\ell-1}((\beta_i(L, a,b))_{i \in [1:\ell-1]})- 2^{\ell-1}\mathbb{1}_{[m+1:b]}(L)) \notag \\
&= (a + 2^{\ell-1}\mathbb{1}_{[m+1:b]}(L)) - a_{\ell-1}((\beta_i(L, a'(L),b'(L)))_{i \in [1:\ell-1]}) \label{eq:a-minus-a-ell-inductive-step}, \\
b - b_\ell((\beta_i(L,a,b)_{i \in [1:\ell]})) &= b - (b_{\ell-1}((\beta_i(L, a,b))_{i \in [1:\ell-1]}) + 2^{\ell-1}(1 - \mathbb{1}_{[m+1:b]}(L))) \notag \\
&= (b - 2^{\ell-1}\mathbb{1}_{[a:m]}(L)) -  b_{\ell-1}((\beta_i(L, a'(L),b'(L)))_{i \in [1:\ell-1]}), \label{eq:b-minus-b-ell-inductive-step}
\end{align}
where as before 
\begin{equation} \label{eq:a-prime-L-b-prime-L-appendix}
(a'(L), b'(L)) = 
\begin{cases} (a, m) & \text{if $L \in [a:m]$}, \textrm{ and} \\
(m +1 , b) & \text{otherwise}.
\end{cases} 
\end{equation}
for $m = \floor{(a+b)/2}$.

One can write $(a'(L), b'(L))$ more explicitly as
\begin{align}
a'(L) &= a + 2^{\ell-1} \mathbb{1}_{[m+1:b]}(L) \, ,\label{eq:a'(L)-explicit}\\ 
b'(L) &= b - 2^{\ell-1}\mathbb{1}_{[a:m]}(L) .\label{eq:b'(L)-explicit}
\end{align}
Indeed, if $L \in [a:b]$, then $a'(L) = a,\; b'(L) = m = \floor{\frac{b+a}{2}}$. But, $b-a+1 = 2^\ell$, so $b+a$ is odd, $a = b - 2^\ell +1$, and therefore
\begin{equation}
m = \frac{b+a-1}{2} = \frac{b + (b-2^\ell +1) -1 }{2} = b - 2^{\ell-1} .\notag 
\end{equation}

The equivalent computation for $L \in [m+1:b]$ establishes \eqref{eq:a'(L)-explicit} and \eqref{eq:b'(L)-explicit}.
Comparing with \eqref{eq:a-minus-a-ell-inductive-step} and \eqref{eq:b-minus-b-ell-inductive-step}, we see 
\begin{align}
a - a_\ell((\beta_i(L,a,b))_{i \in [1:\ell]}) &= a'(L) - a_{\ell-1}((\beta_i(L,a'(L), b'(L)))_{i \in [1:\ell-1]}), \notag \\ 
b - b_\ell((\beta_i(L,a,b))_{i \in [1:\ell]}) &=b'(L) - b_{\ell-1}((\beta_i(L,a'(L), b'(L)))_{i \in [1:\ell-1]}). \notag 
\end{align}
Hence, by the inductive hypothesis:
\begin{equation}
a - a_\ell((\beta_i(L,a,b))_{i \in [1:\ell]}) = L \quad \text{and } \quad b - b_\ell((\beta_i(L, a, b))_{i \in [1:\ell]}) = L \notag 
\end{equation}
as required.
\end{proof}

\subsection{Lemmas related to \texorpdfstring{$\ell_{min}$}{ellmin}}

The central lemma of this section is 
\begin{lemma} \label{lem:j-min-symmetry}
For
$L \in [a_T : b_T]$,
\begin{align} \label{eq:j-min-symmetry}
\ell_{min}(\theta, \rho, B, h, R) &= \ell_{min}(\Phi_{h/R}^{RL}(\theta, \rho), B^*,h, R).
\end{align}
\end{lemma}

\begin{figure}
\begin{tikzpicture} [
    level 1/.style={sibling distance = 8cm, level distance = 2cm},
    level 2/.style={sibling distance = 4cm},
    level 3/.style={sibling distance = 2cm},
    emph/.style={edge from parent/.style={red, very thick, draw}},
    norm/.style={edge from parent/.style={black, thin, draw}}
    ]

    \node{$\{-3,-2,-1,0,1,2,3,4\}$}
    child[norm] { node {$\{-3,-2,-1,0\}$}
        child[norm] { node {$\{-3,-2\}$}
            child { node {-3} edge from parent node[midway, left] {0} }
            child { node {-2} edge from parent node[midway, right] {1} }
            edge from parent node[midway, left]{0}
        }
        child[norm] { node {$\{-1,0\}$}
            child { node {-1} edge from parent node[midway, left] {0} }
            child { node {0} edge from parent node[midway, right] {1} }
            edge from parent node[midway, right]{1}
        }
        edge from parent node[midway, above]{0}
    }
    child[norm] { node {\textcolor{red}{$\{1,2,3,4\}$}}
        child[emph] { node {\textcolor{red}{$\{1,2\}$}}
            child[emph] { node {\textcolor{red}{1}} edge from parent node[midway, left] {0} }
            child[emph] { node {\textcolor{red}{2}} edge from parent node[midway, right] {1} }
            edge from parent node[midway, left]{0}
        }
        child[emph] { node {\textcolor{red}{$\{3,4\}$}}
            child { node {\textcolor{red}{3}} edge from parent node[midway, left] {0} }
            child { node {\textcolor{red}{4}} edge from parent node[midway, right] {1} }
            edge from parent node[midway, right]{1}
        }
        edge from parent node[midway, above] {1}
    };

\end{tikzpicture}
\caption{$\mathcal{T}([1:4])$ highlighted inside  $\mathcal{T}([-3:4])$ }
\end{figure}
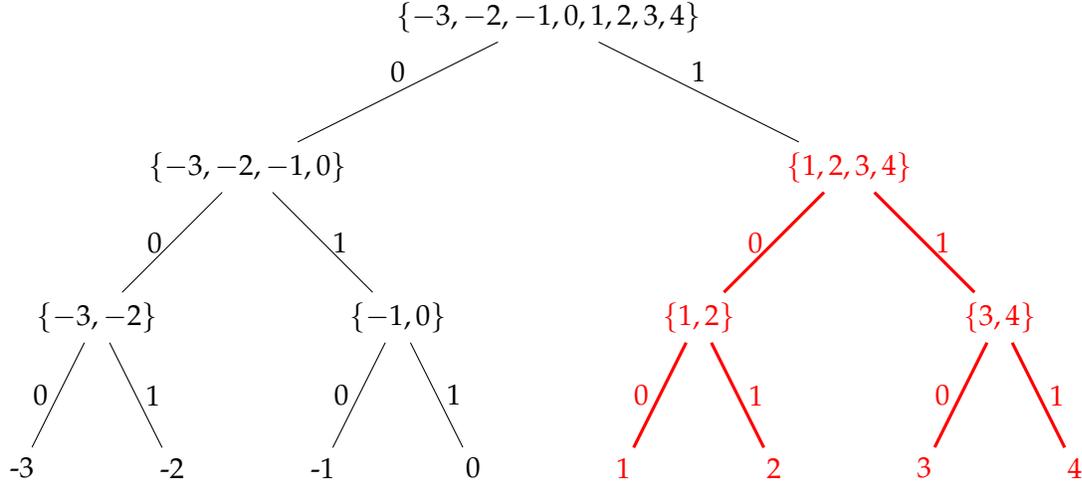

To aid understanding of the proof of Lemma \ref{lem:j-min-symmetry} we provide a brief heuristic explanation. When executing Algorithm \ref{algo:leapfrog-orbit-selection} from $(\theta, \rho)$ with sequence $B$, the algorithm checks many intervals $[a:b]$ for U-turns against $(\theta, \rho)$. This is either done explicitly with $\mathbb{1}_{\text{U-turn}}(\theta, \rho, a, b, h, R)$ or implicitly through $\mathbb{1}_{\text{sub-U-turn}}(\theta, \rho, \tilde{a}, \tilde{b}, h, R)$ for some $[\tilde{a} : \tilde{b}]$ such that $[a:b]$ is a vertex in the binary tree implicitly defined by $[\tilde{a} : \tilde{b}]$.

When executing Algorithm \ref{algo:leapfrog-orbit-selection} from $\Phi_{h/R}^{RL}(\theta, \rho)$ with sequence $B^*$ we want to show that every interval $[a^*:b^*]$ which is implicitly or explicitly checked for a U-turn against $\Phi_{h/R}^{RL}(\theta, \rho)$ is an appropriately shifted version of some interval which was checked when executing from $(\theta, \rho)$. Doing so will let us show that each computation terminates after the same number of doublings $\ell$. However, some of the intervals which are checked explicitly in one computation are checked implicitly in the other and vice-versa, which makes direct comparison difficult.  

To work around this difficulty, we simply characterize the set of all intervals which need to be checked during each execution. Doing so will require additional notation not needed elsewhere. For a set $[a:b]$ with $b-a+1 = 2^\ell$, define a nested sequence of sets as follows. 
\begin{align}
\mathcal{T}_\ell &= \{ [a : b] \}, \notag\\
\mathcal{T}_{j - 1}([a:b]) &= \{ [\tilde{a}: \tilde{b}] \mid \tilde{a} = a', \tilde{b} = m \quad \text{ or } \quad \tilde{a} = m+1, \tilde{b} = b' \quad \\
& \hspace{3cm} \text{for some} \quad [a':b'] \in \mathcal{T}_{j}([a:b]), \; m = \floor{(a+b)/2}\}, \notag \\ 
\mathcal{T}([a:b]) &= \bigcup_{j=1}^\ell \mathcal{T}_j([a:b]). \notag
\end{align}
The set $\mathcal{T}([a:b])$ essentially represents the binary tree corresponding to the interval $[a:b]$. A proof by induction immediately shows 
\begin{equation} \label{eq:descendants-downward-closed}
\mathcal{T}([a:b]) = \mathcal{T}([a:m]) \cup \mathcal{T}([m+1 : b]) \cup \{ [a:b]\}.
\end{equation}
We will say that $\mathcal{T}([a:b])$ is \emph{sub-U-turn free} if 
\begin{equation} \label{eq:def:sub-u-turn-free}    
\forall [\tilde{a} : \tilde{b}] \in \mathcal{T}([a:b]), \; \; \mathbb{1}_{\text{U-turn}}(\theta, \rho, a, b, h, R) = 0 .
\end{equation}
This terminology is motivated by Lemma \ref{lem:sub-u-turn-free-equivalence} which shows that $\mathcal{T}([a:b])$ is sub-U-turn free if and only if $\mathbb{1}_{\text{sub-U-turn}}(\theta, \rho, a, b, h, R) = 0 $. 

\begin{proof}[Proof of Lemma \ref{lem:j-min-symmetry}]
Let $\ell_{min} = \ell_{min}(\theta, \rho, B, h, R)$ and $\ell_{min}^* = \ell_{min}(\Phi_{h/R}^{RL}(\theta, \rho, B^*, h, R)$. We separately show $\ell_{min}^* \leq \ell_{min}$ and $\ell_{min}^* \geq \ell_{min}$.  

\underline{\textbf{Step 1}}: $\ell_{min}^* \leq \ell_{min}.$

If $\ell_{min} = M$ we are done. Otherwise, we need to show 
\begin{multline} \label{eq:j-min-upper-bound}
\max \big( \mathbb{1}_{\text{U-turn}}(\Phi_{h/R}^{RL}(\theta, \rho), a_{\ell_{min}}((B_i^*)_{i \in [1:\ell_{min}]}), b_{\ell_{min}}((B_i^*)_{i \in [1:\ell_{min}]}), h, R), \\ \mathbb{1}_{\text{sub-U-turn}}(\Phi_{h/R}^{RL}(\theta, \rho), \tilde{a}_{\ell_{min}+1}((B_i^*)_{i \in [1:\ell_{min}+1]}), \tilde{b}_{\ell_{min}+1}((B_i^*)_{i \in [1:\ell_{min}+1]}), h, R) \big) = 1 .
\end{multline}

By Lemma \ref{lem:initial-point-symmetry}
\begin{align}
a_{\ell_{min}}((B_i^*)_{i \in [1:\ell_{min}]}) &= a_T - L \, , \notag \\
b_{\ell_{min}}((B_i^*)_{i \in [1:\ell_{min}]}) &= b_T - L  \, ,\notag \\
\tilde{a}_{\ell_{min}+1}((B_i)_{i \in [1: \ell_{min} +1]}) &= a_{\ell_{min}}((B_i^*)_{i \in [1:\ell_{min}]}) + (-1)^{B_{\ell_{min}+1}} \cdot 2^{\ell_{min}} \notag \\
&= (a_T - L) + (-1)^{B_{\ell_{min} +1}} \cdot 2^{\ell_{min}} \notag \\
&= \tilde{a}_{\ell_{min} + 1}((B_i)_{i \in [1: \ell_{min} + 1]}) - L \, ,\notag \\
\tilde{b}_{\ell_{min}+1}((B_i)_{i \in [1: \ell_{min} +1]}) &= b_{\ell_{min}}((B_i^*)_{i \in [1:\ell_{min}]}) + (-1)^{B_{\ell_{min}+1}} \cdot 2^{\ell_{min}} \notag \\
&= (b_T - L) + (-1)^{B_{\ell_{min} +1}} \cdot 2^{\ell_{min}} \notag \\
&= \tilde{b}_{\ell_{min} + 1}((B_i)_{i \in [1: \ell_{min} + 1]}) - L .\notag 
\end{align}
Hence, using \eqref{eq:u-turn-symmetry}
\begin{align} 
&\max \big( \mathbb{1}_{\text{U-turn}}(\Phi_{h/R}^{RL}(\theta, \rho), a_{\ell_{min}}((B_i^*)_{i \in [1:\ell_{min}]}), b_{\ell_{min}}((B_i^*)_{i \in [1:\ell_{min}]}), h, R), & \notag \\
& & \hspace{-15cm} \mathbb{1}_{\text{sub-U-turn}}(\Phi_{h/R}^{RL}(\theta, \rho), \tilde{a}_{\ell_{min}+1}((B_i^*)_{i \in [1:\ell_{min}+1]}), \tilde{b}_{\ell_{min}+1}((B_i^*)_{i \in [1:\ell_{min}+1]}), h, R) \big) =  \notag \\
&\max \big( \mathbb{1}_{\text{U-turn}}(\Phi_{h/R}^{RL}(\theta, \rho), a_T - L , b_T - L, h, R), & \notag \\ 
& &\hspace{-15cm} \mathbb{1}_{\text{sub-U-turn}}(\Phi_{h/R}^{RL}(\theta, \rho), \tilde{a}_{\ell_{min}+1}((B_i)_{i \in [1:\ell_{min}]}) - L, \tilde{b}_{\ell_{min}+1}((B_i)_{i \in [1:\ell_{min}]}) - L, h, R) \big) = \notag \\
&\max \big( \mathbb{1}_{\text{U-turn}}(\theta, \rho , a_T  , b_T , h, R), \mathbb{1}_{\text{sub-U-turn}}(\theta, \rho, \tilde{a}_{\ell_{min}+1}((B_i)_{i \in [1:\ell_{min}]}), \tilde{b}_{\ell_{min}+1}((B_i)_{i \in [1:\ell_{min}]}) , h, R) \big) = 1 . \notag 
\end{align}

\underline{\textbf{Step 2:}}  $\ell_{min}^* \geq \ell_{min}.$

We need to show for $\ell < \ell_{min}$
\begin{align} 
\mathbb{1}_{\text{U-turn}}(\Phi_{h/R}^{RL}(\theta, \rho), a_{\ell}((B_i^*)_{i \in [1:\ell]}), b_{\ell}((B_i^*)_{i \in [1:\ell]}), h, R) &= 0 \, ,\label{eq:j-min-lower-bound-u-turn} \\
\mathbb{1}_{\text{sub-U-turn}}(\Phi_{h/R}^{RL}(\theta, \rho), \tilde{a}_{\ell+1}((B_i^*)_{i \in [1:\ell+1]}), \tilde{b}_{\ell+1}((B_i^*)_{i \in [1:\ell+1]}), h, R) &= 0 .\label{eq:j-min-lower-bound-sub-u-turn}
\end{align}

By Lemma \ref{lem:strict-descendants-of-top-sub-u-turn-free}, we know for $\ell < \ell_{min}$ that 
\[\mathcal{T}([a_\ell((B_i^*)_{i \in [1:\ell]}) + L : b_\ell((B_i^*)_{i \in [1:\ell]}) + L ]) \cup \mathcal{T}([\tilde{a}_{\ell+1}((B_i^*)_{i \in [1:\ell+1]}) + L : \tilde{b}_{\ell + 1}((B_i^*)_{i \in [1:\ell + 1]}) +L ])  \]
is sub-U-turn free. Hence,
\begin{align}
\mathbb{1}_{\text{U-turn}}(\theta, \rho, a_{\ell}((B_i^*)_{i \in [1: \ell]}) + L,  b_{\ell}((B_i^*)_{i \in [1: \ell]}) + L, h, R) &= 0 \, , \label{eq:shifted-no-U-turn} \\
\mathbb{1}_{sub-U-turn}(\theta, \rho,\tilde{a}_{\ell+1}((B_i)_{i \in [1: \ell+1]}) + L , \tilde{b}_{\ell+1}((B_i)_{i \in [1: \ell+1]})+ L, h, R)  &= 0 .\label{eq:shifted-no-sub-U-turn}
\end{align} 
where we used Lemma \ref{lem:sub-u-turn-free-equivalence} to establish \eqref{eq:shifted-no-sub-U-turn}. Using the shift invariance \eqref{eq:u-turn-symmetry} of $\mathbb{1}_{\text{U-turn}}$ and $\mathbb{1}_{\text{sub-U-turn}}$ we obtain \eqref{eq:j-min-lower-bound-u-turn} and \eqref{eq:j-min-lower-bound-sub-u-turn}. 

Combining these two steps above implies 
\[\ell_{min}(\Phi_{h/R}^{RL}(\theta, \rho), B^*, h, R)  = \ell_{min}(\theta, \rho, B, h, R)\]
as desired.
\end{proof}

\begin{lemma} \label{lem:sub-u-turn-free-equivalence}
For $b -a +1 = 2^\ell$
\begin{equation} \label{eq:descendants-sub-u-turn-a-b}
    \mathbb{1}_{\text{sub-U-Turn}}(\theta, \rho, a, b, h, R) = 0,
\end{equation}
if and only if 
\begin{equation} \label{eq:all-descendents-u-turn-free}
\forall [\tilde{a} : \tilde{b}] \in \mathcal{T}([a:b]), \mathbb{1}_{\text{U-turn}}(\theta, \rho, a, b, h, R) = 0.
\end{equation}
\end{lemma}

\begin{remark}
Lemma \ref{lem:sub-u-turn-free-equivalence} simply establishes that $\mathbb{1}_{sub-U-turn}(\theta, \rho, a, b, h, R)$ checks in a literal sense that no vertex in the associated tree has a U-turn.
\end{remark}

\begin{proof}
We prove this by induction on $\ell$. First, suppose $\ell = 0$. In this case, $b-a +1 = 1$ and hence $b=a$ and the claim is vacuously true. 

Suppose next the claim is true for $\ell' < \ell$. Suppose \eqref{eq:descendants-sub-u-turn-a-b} is true. We then have by definition 
\begin{align}
\mathbb{1}_{U-turn}(\theta, \rho, a, b, h, R) &= 0  \, ,\label{eq:u-turn-descendants-a-b-root-node} \\
\mathbb{1}_{sub-U-turn}(\theta, \rho, a, m, h, R) &= 0 \, ,\label{eq:sub-u-turn-induction-a-m}\\
\mathbb{1}_{sub-U-turn}(\theta, \rho, m+1, b, h, R) &= 0 \, ,\label{eq:sub-u-turn-induction-mp1-b}
\end{align}
for $m = \floor{(a+b)/2}$.
By \eqref{eq:sub-u-turn-induction-a-m} and \eqref{eq:sub-u-turn-induction-mp1-b}, 
\begin{equation} \label{eq:strict-descendants-u-turn-free}
\forall [\tilde{a} : \tilde{b}] \in \mathcal{T}([a:m]) \cup \mathcal{T}([m+1:b]), \mathbb{1}_{U-turn}(\theta, \rho, \tilde{a}, \tilde{b}) = 0.
\end{equation}
Finally, using \eqref{eq:u-turn-descendants-a-b-root-node} we have \[\forall [\tilde{a} : \tilde{b}] \in \mathcal{T}([a:m]) \cup \mathcal{T}([m+1:b]) \cup \{ [a:b] \}, \; \; \mathbb{1}_{U-turn}(\theta, \rho, \tilde{a}, \tilde{b}) = 0 .\]
Hence, \eqref{eq:descendants-sub-u-turn-a-b} implies \eqref{eq:all-descendents-u-turn-free}. Going the other way, suppose \eqref{eq:all-descendents-u-turn-free} is true. In this case, as $[a:b] \in \mathcal{T}([a:b])$ we get \eqref{eq:u-turn-descendants-a-b-root-node} directly. Additionally, as \eqref{eq:descendants-downward-closed} and \eqref{eq:all-descendents-u-turn-free} clearly imply \eqref{eq:strict-descendants-u-turn-free}, which implies \eqref{eq:sub-u-turn-induction-a-m} and \eqref{eq:sub-u-turn-induction-mp1-b} by induction. From the definition of $\mathbb{1}_{sub-U-turn}(\theta, \rho, a,b, h, R) = 0$ and thus (\ref{eq:all-descendents-u-turn-free}) implies (\ref{eq:descendants-sub-u-turn-a-b}).

\end{proof}

\begin{lemma} \label{lem:beta-downward-closed}
For $(B, \ell, a, b, L) \in \mathcal{D}$
\begin{equation} \label{eq:beta-downward-closed}
\mathcal{T}([a_{j-1}((B_i^*)_{i \in [1:j-1]}) +L : b_{j-1}((B_i^*)_{i \in [1:j-1]}) + L ]) \subseteq \mathcal{T}([a_{j}((B_i^*)_{i \in [1:j]}) + L : b_{j}((B_i^*)_{i \in [1:j]}) + L]) .
\end{equation}
\end{lemma}

\begin{remark}
Lemma \ref{lem:beta-downward-closed} and Corollary \ref{cor:containment-of-partial-trees} establish that the intervals which need to be checked for U-turns when running Algorithm \ref{algo:leapfrog-orbit-selection} with the transformed sequence $B^*$ are all shifted versions of the intervals which need to be checked when running Algorithm \ref{algo:leapfrog-orbit-selection} with the original sequence $B$.
\end{remark}

\begin{proof} 
Proceed by induction on $\ell$. If $\ell = 0$, the claim is clearly vacuously true. 

For $\ell >0$, first note that by definition 
\[(B^*_i)_{i \in [1:\ell- 1]} = (\beta(L, a'(L), b'(L)))_{i \in [1: \ell-1]} \]
for $(a'(L), b'(L))$ defined in (\ref{eq:a-prime-L-b-prime-L-appendix}). Hence, employing Lemma \ref{lem:initial-point-symmetry} we find 
\[(a'(L), b'(L)) = (a_{\ell-1}((B_i^*)_{i \in [1:\ell -1]}) + L , b_{\ell}((B_i^*)_{i \in [1:\ell -1]} ) + L ).\]
In general, $[a'(L) : b'(L)]$ is either $[a:m]$ or $[m+1:b]$ so by (\ref{eq:descendants-downward-closed}), 
\begin{align}
\mathcal{T}([a_{\ell - 1}((B_i^*)_{i \in [1 :\ell-1]}) + L : b_{\ell -1}((B_i^*)_{i \in [1:\ell -1]}) + L]) &= \mathcal{T}([a'(L): b'(L)] \notag \\
&\subseteq  \mathcal{T}([a:b]) \notag \\
&= \mathcal{T}([a_{\ell}((B_i^*)_{i \in [1:\ell]}) + L : b_{\ell}((B_i^*)_{i \in [1:\ell]}) + L]). \notag 
\end{align}
Hence, (\ref{eq:beta-downward-closed}) holds for $j = \ell$. 

But, for $j \leq \ell-1$, $(B_i^*)_{i \in [1:j]}(L, a, b,B) =(B_i^*)_{i \in [1:j]}(L, a'(L), b'(L), B)$. Therefore, since $b'(L) - a'(L) + 1 = 2^{\ell -1}$, (\ref{eq:beta-downward-closed}) holds for $j \in [1:\ell-1]$ by the inductive hypothesis. Combining these, (\ref{eq:beta-downward-closed}) holds for $j \in [1:\ell]$ as desired.

\end{proof}

\begin{corollary} \label{cor:containment-of-partial-trees}
For $b - a +1 = 2^\ell$, $L \in [a : b]$, $0 \leq j \leq \ell$
\[\mathcal{T}([a_{j}((B_i^*)_{i \in [1:j]}) + L : b_j((B_i^*)_{i \in [1:j]}) + L]) \subseteq \mathcal{T}([a: b]) .\]
\end{corollary}

\begin{proof}
This follows immediately by combining Lemmas \ref{lem:beta-downward-closed} and \ref{lem:initial-point-symmetry}.
\end{proof}

\begin{lemma} \label{lem:strong-minimality-of-ell-min}
Let $\ell_{min} = \ell_{min}(\theta, \rho, B, h, R)$. For $\ell < \ell_{min}$, $\mathcal{T}([a_{\ell}((B_i)_{i \in [1:\ell]}): b_{\ell}((B_i)_{i \in [1:\ell]})])$ is sub-U-turn free.
\end{lemma}

\begin{remark}
Lemma \ref{lem:strong-minimality-of-ell-min} establishes that none of the intervals in $\mathcal{T}([a_T:b_T])$ has a U-turn.
\end{remark}

\begin{proof}
We prove this by induction on $\ell$. For $\ell = 0$, the claim is trivially true from 
\[\mathcal{T}([a_0 : b_0]) = \{ [0] \}\]
and 
\[\mathbb{1}_{\text{U-turn}}(\theta, \rho, a,a, h, R) = 0. \]

For $\ell >0$, note that 
\begin{align*} 
& [a_\ell((B_i)_{i \in [1:\ell]}) : b_\ell((B_i)_{i \in [1:\ell]})] \\
& \quad = [a_{\ell-1}((B_i)_{i \in [1:\ell -1 ]}) : b_{\ell -1}((B_i)_{i \in [1:\ell - 1 ]})] \cup [\tilde{a}_\ell((B_i)_{i \in [1:\ell]}) : \tilde{b}_\ell((B_i)_{i \in [1:\ell]})] \notag 
\end{align*}
with one of these sets being the upper half and the other set being the lower half. Hence, by \eqref{eq:descendants-downward-closed} 
\begin{multline}
\mathcal{T}([a_\ell((B_i)_{i \in [1:\ell]}) : b_\ell((B_i)_{i \in [1:\ell]})]) \\ = \mathcal{T}([a_{\ell-1}((B_i)_{i \in [1:\ell -1 ]}) : b_{\ell -1}((B_i)_{i \in [1:\ell - 1 ]})]) \cup \mathcal{T}([\tilde{a}_\ell((B_i)_{i \in [1:\ell]}) : \tilde{b}_\ell((B_i)_{i \in [1:\ell]})])  \\\cup \{[a_\ell((B_i)_{i \in [1:\ell]}) : b_\ell((B_i)_{i \in [1:\ell]})]\}. \notag 
\end{multline}
By the inductive hypothesis, $\mathcal{T}([a_{\ell-1}((B_i)_{i \in [1:\ell -1 ]}) : b_{\ell -1}((B_i)_{i \in [1:\ell - 1 ]})])$ is sub-U-turn free. Since ${\ell < \ell_{min}}$ we find 
\begin{align}
\mathbb{1}_{\text{U-turn}}(\theta, \rho, a_{\ell}((B_i)_{i \in [1:\ell]}), b_{\ell}((B_i)_{i \in [1:\ell]})) &= 0 \, ,\notag \\
\mathbb{1}_{\text{sub-U-turn}}(\theta, \rho, \tilde{a}_{\ell}((B_i)_{i \in [1: \ell]}), \tilde{b}_{\ell}((B_i)_{i \in [1: \ell]})) &= 0. \notag 
\end{align}
Hence, $\mathcal{T}([\tilde{a}_\ell((B_i)_{i \in [1:\ell]}) : \tilde{b}_\ell((B_i)_{i \in [1:\ell]})])$ and $\{[a_\ell((B_i)_{i \in [1:\ell]}) : b_\ell((B_i)_{i \in [1:\ell]})]\}$ are both sub-U-turn free. Therefore, the union of these three sets is sub-U-turn free and the claim follows by induction.  
\end{proof}

\begin{lemma} \label{lem:strict-descendants-of-top-sub-u-turn-free}
Let $\ell_{min} = \ell_{min}(\theta, \rho, B, h, R)$ and $L \in [a_T:b_T]$. 

Then, for all $\ell < \ell_{min}$
\begin{equation} \label{eq:tree-and-proposal-sub-u-turn-free-below-top}
\mathcal{T}([a_{\ell}((B_i^*)_{i \in [1:\ell]}) + L: b_{\ell}((B_i^*)_{i \in [1:\ell]}) + L]) \cup \mathcal{T}([\tilde{a}_{\ell+1}((B_i^*)_{i \in [1:\ell+1]}) + L: \tilde{b}_{\ell+1}((B_i^*)_{i \in [1:\ell+1]}) + L]) 
\end{equation} 
is sub-U-turn free.

\end{lemma} 

\begin{remark}
Lemma \ref{lem:strict-descendants-of-top-sub-u-turn-free} transfers the results from Lemma \ref{lem:strong-minimality-of-ell-min} from the original sequence of sets $[a_\ell((B_i)_{i \in [1:\ell]}): b_{\ell}((B_i)_{i \in [1:\ell]})]$ to the sequence of sets $[a_{\ell}((B_i^*)_{i \in [1: \ell]}) + L : b_{\ell}((B_i^*)_{i \in [1:\ell]}) + L]$. The above in fact slightly strengthens the results coming from Lemma \ref{lem:strong-minimality-of-ell-min} in order to prove Lemma \ref{lem:j-min-symmetry} cleanly.

Using the shift invariance of $\mathbb{1}_{\text{U-turn}}$ and $\mathbb{1}_{\text{sub-U-turn}}$ in \eqref{eq:u-turn-symmetry}, the lemma is equivalent to saying that \texttt{leapfrog-orbit-selection} started from $\Phi_{h/R}^{RL}(\theta, \rho)$ with sequence $B^*$ does not encounter a U-turn or sub-U-turn for $\ell < \ell_{min}(\theta, \rho, B, h, R)$. 

In the proof of Lemma \ref{lem:j-min-symmetry}, this is essential for establishing
\[\ell_{min}(\Phi_{h/R}^{RL}(\theta, \rho), B^*, h, R) \geq \ell_{min}(\theta, \rho, B, h, R).\]
\end{remark}

\begin{proof}
The lemma follows from establishing that $\mathcal{T}([a_T :b_T]) \setminus \{ [a_T: b_T ]\}$ is sub-U-turn free.

Decompose the set $[a_T : b_T]$ as 
\[[a_T : b_T] = [a_{\ell_{min}-1}((B_i)_{i \in [1:\ell_{min}-1 ]}): b_{\ell_{min}-1}((B_i)_{i \in [1:\ell_{min}-1]})] \cup [\tilde{a}_{\ell_{min}}((B_i)_{i \in [1: \ell_{min}]}) : \tilde{b}_{\ell_{min}}((B_i)_{i \in [1:\ell_{min}]})].\]
By assumption $\mathbb{1}_{sub-U-turn}(\theta, \rho,\tilde{a}_{\ell_{min}}((B_i)_{i \in [1: \ell_{min}]}),  \tilde{b}_{\ell_{min}}((B_i)_{i \in [1:\ell_{min}]}), h, R) = 0$, so by Lemma \ref{lem:sub-u-turn-free-equivalence} the set 
\[ \mathcal{T}( [\tilde{a}_{\ell_{min}}((B_i)_{i \in [1: \ell_{min}]}) : \tilde{b}_{\ell_{min}}((B_i)_{i \in [1:\ell_{min}]})])\] is sub-U-turn free. Additionally, by Lemma \ref{lem:strong-minimality-of-ell-min} the set  \[\mathcal{T}([a_{\ell_{min}-1}((B_i)_{i \in [1:\ell_{min}-1]}): b_{\ell_{min}-1}((B_i)_{i \in [1:\ell]})])\] is sub-U-turn free. 
Therefore, by \eqref{eq:descendants-downward-closed} the set $ \mathcal{T}([a_T:b_T]) \setminus \{ [a_T :b_T]\}$ is sub-U-turn free. 

To establish \eqref{eq:strict-descendants-u-turn-free} write for $\ell < \ell_{min}$
\begin{multline}     
[a_{\ell+1}((B_i^*)_{i \in [1: \ell +1]}) +L : b_{\ell+1}((B_i^*)_{i \in [1:\ell+1]}) + L] = \\ [a_{\ell}((B_i^*)_{i \in [1: \ell]}) +L : b_{\ell}((B_i^*)_{i \in [1:\ell]}) + L] \cup  [\tilde{a}_{\ell+1}((B_i^*)_{i \in [1: \ell +1]}) +L : \tilde{b}_{\ell+1}((B_i^*)_{i \in [1:\ell+1]}) + L].
\end{multline}
By Corollary \ref{cor:containment-of-partial-trees} 
\[\mathcal{T}([a_{\ell+1}((B_i^*)_{i \in [1: \ell +1]}) +L : b_{\ell+1}((B_i^*)_{i \in [1:\ell+1]}) + L]) \subseteq \mathcal{T}([a_T:b_T]).\]
Hence, by \eqref{eq:descendants-downward-closed} 
\begin{multline}    
\mathcal{T}([a_{\ell}((B_i^*)_{i \in [1: \ell]}) +L : b_{\ell}((B_i^*)_{i \in [1:\ell]}) + L]) \\  \cup \mathcal{T} ([\tilde{a}_{\ell+1}((B_i^*)_{i \in [1: \ell +1]}) +L : \tilde{b}_{\ell+1}((B_i^*)_{i \in [1:\ell+1]}) + L]) \subseteq \mathcal{T}([a_T: b_T]).
\end{multline}
Additionally, 
\begin{multline}
[a_T:b_T] \not \in \mathcal{T}([a_{\ell}((B_i^*)_{i \in [1: \ell]}) +L : b_{\ell}((B_i^*)_{i \in [1:\ell]}) + L]) \\  \cup \mathcal{T} ([\tilde{a}_{\ell+1}((B_i^*)_{i \in [1: \ell +1]}) +L : \tilde{b}_{\ell+1}((B_i^*)_{i \in [1:\ell+1]}) + L]).
\end{multline}
Thus, as 
\[\mathcal{T}([a_T; b_T]) \setminus \{ [a_T:b_T]\}\]
is sub-U-turn free, we find
\begin{multline}
\mathcal{T}([a_{\ell}((B_i^*)_{i \in [1: \ell]}) +L : b_{\ell}((B_i^*)_{i \in [1:\ell]}) + L]) \\  \cup \mathcal{T} ([\tilde{a}_{\ell+1}((B_i^*)_{i \in [1: \ell +1]}) +L : \tilde{b}_{\ell+1}((B_i^*)_{i \in [1:\ell+1]}) + L])
\end{multline}
is sub-U-turn free as desired.

\end{proof}

\printbibliography

\end{document}